\documentclass[conference]{IEEEtran}
\IEEEoverridecommandlockouts
\usepackage{amsmath,amssymb,amsthm,mathtools}
\usepackage{array,booktabs,cite,enumitem,microtype,tikz}
\usepackage{xcolor}
\usetikzlibrary{arrows.meta,positioning}
\usepackage[hidelinks]{hyperref}
\newtheorem{theorem}{Theorem}
\newtheorem{lemma}{Lemma}
\newtheorem{proposition}{Proposition}

\newcommand{\Rge}{\mathbb R_{\ge0}}
\newcommand{\cH}{\mathcal H}
\newcommand{\cA}{\mathcal A}
\newcommand{\cB}{\mathcal B}
\newcommand{\cI}{\mathcal I}
\newcommand{\cJ}{\mathcal J}
\newcommand{\src}{\mathrm{src}}
\newcommand{\dst}{\mathrm{dst}}
\newcommand{\NC}{\mathrm{NC}}
\newcommand{\MCF}{\mathrm{MCF}}
\title{Routing Optimality for Demand Graphs\\of Matching Number at Most Two}
\author{\IEEEauthorblockN{Xiaohong Cai and Yulin Chen}\thanks{Xiaohong Cai and Yulin Chen are with The Chinese University of Hong Kong, Hong Kong SAR, China (e-mail: cx021@ie.cuhk.edu.hk; 1155189560@link.cuhk.edu.hk).}}
\begin{document}
\maketitle
\begin{abstract}
We prove that network coding and fractional routing have the same capacity region in undirected networks whose demand graph has matching number at most two, that is, when no three demand edges have pairwise disjoint endpoints. The metric-packing reduction of Liu et al. reduces this problem to a distance-weighted communication inequality on the fixed graph $\Gamma_{3,3}$. We first prove two entropy inequalities for sessions between designated terminal pairs on $\Gamma_{3,3}$. We then extend this result to sessions with arbitrary endpoints to obtain the distance-weighted communication bound required by the reduction. The proof combines conditional entropy inequalities with two graph contractions and an information inequality for deterministic interaction. Combining this inequality with the existing reduction establishes routing optimality for this class of demand graphs.
\end{abstract}

\section{Introduction}
Network coding allows a node to transmit functions of previously received messages \cite{Ahlswede2000,Koetter2003}. Li and Li \cite{LiLi2004} and Harvey, Kleinberg, and Lehman \cite{Harvey2004} conjectured that network coding and fractional routing have the same capacity region in undirected networks. Recent preprints by Zhang, Li, and Li \cite{ZhangLiLi2026} and Braverman and He \cite{BravermanHe2026} report counterexamples to this conjecture. Although coding can outperform routing in general, it remains important to identify structural conditions under which network coding and fractional routing have the same capacity region. We prove that, for every undirected network and every edge-capacity assignment, network coding and fractional routing have the same capacity region whenever no three demand pairs have pairwise disjoint endpoints.

The cut condition bounds the total rate across each partition of the vertices into two sets by the capacity of the edges crossing it. It characterizes routing for two sessions \cite{Hu1963}, at most four terminal locations \cite{Seymour1980}, and certain planar instances \cite{OkamuraSeymour1981}. When cuts do not characterize routing, stronger information inequalities can still rule out a coding advantage \cite{Jain2006,Harvey2006,Yin2018}. Metric packing results \cite{Karzanov1990,Karzanov1994,Hirai2010} extend this approach by retaining shortest-path information beyond cuts.

Liu, Que, Li, and Li \cite{LiuQueLiLi2026} use metric packings to reduce coding--routing comparisons to communication inequalities on fixed graphs. For demand graphs of matching number at most two, their reduction requires communication inequalities on four fixed graphs: $K_2$, $K_{2,3}$, $K_{3,3}$, and $\Gamma_{3,3}$. Here $K_2$ is the two-vertex complete graph, $K_{p,q}$ is a complete bipartite graph, and $\Gamma_{3,3}$ is defined in Section~\ref{sec:gamma}. The required communication inequalities on $K_2$, $K_{2,3}$, and $K_{3,3}$ have been established \cite{Yin2018,LiuQueLiLi2026}, leaving the corresponding inequality on $\Gamma_{3,3}$ to be proved. We establish this distance-weighted communication inequality for arbitrary sessions on $\Gamma_{3,3}$ and apply the reduction of \cite{LiuQueLiLi2026} to obtain our main theorem.

We first establish two directional entropy bounds for sessions between distinct outer terminals in the same vertex class of $\Gamma_{3,3}$. The forward bound uses reconstruction of transmitted messages and the entropy chain rule. The reverse bound combines two graph contractions with an information inequality for deterministic interaction. We then extend arbitrary sessions along shortest paths between outer terminals to derive the distance-weighted communication inequality on $\Gamma_{3,3}$.

The paper is organized as follows. Section~\ref{sec:model} presents the network model and states the main theorem. Section~\ref{sec:gamma} introduces $\Gamma_{3,3}$ and states the directional entropy bounds, which are proved in Section~\ref{sec:proof}. Section~\ref{sec:extension} extends these bounds to arbitrary sessions and completes the proof of the main theorem. Section~\ref{sec:conclusion} concludes the paper. The full proofs are provided in the supplementary material following the references.

\section{Network Model and Main Result}
\label{sec:model}
Following the session notation of \cite[Sec.~III-A]{LiuQueLiLi2026}, we consider a \emph{multiple-unicast network model} $(G,c,\cH)$. Here $G=(V,E)$ is a finite undirected graph, called the \emph{supply graph}; $c:E\to\Rge$ is an \emph{edge-capacity assignment}; and $\cH=(h_1,\ldots,h_k)$ is a collection of \emph{unicast sessions}. For each edge $e\in E$, $c(e)$ is its capacity, shared by transmissions in both directions. Each session $h$ requires a \emph{source node} $s_h$ to deliver a message $X_h$ to a \emph{destination node} $t_h\ne s_h$, and the messages $X_1,\ldots,X_k$ are mutually independent. Let each source message $X_h$ be uniformly distributed over a finite alphabet $\mathcal X_h$. A vertex that is the source or the destination of at least one session is called a \emph{terminal}. Different sessions may share terminals, and two sessions may even have the same source and the same destination.

A protocol consists of $N$ transmissions, indexed by $\tau=1,\ldots,N$ in an order consistent with their causal dependencies. Simultaneous transmissions may be indexed in either order. The sender, receiver, and message length of each transmission are fixed in advance. Each message is a deterministic function of the sender's local source messages and previously received messages. The protocol is \emph{zero-error} if every destination $t_h$ recovers $X_h$ for every value of the source messages $X_1,\dots,X_k$. For a directed edge $(u,v)$, let $Y_{uv}$ denote the complete ordered transcript that $u$ sends to $v$, and let $b_{uv}$ be its fixed length in bits.
Let $r=(r_h)_{h\in\mathcal H}\in\Rge^k$, where $r_h$ is the target rate of session $h$. The rate vector $r$ is \emph{supported} under the capacities $c$ if there exist a deterministic zero-error protocol and a positive integer $n$ such that

\begin{align}
\log_2|\mathcal X_h|&\ge nr_h &&\text{for every }h\in\mathcal H,\label{eq:source-rate-budget}\\
b_{uv}+b_{vu}&\le nc(uv) &&\text{for every }uv\in E.\label{eq:edge-capacity-budget}
\end{align}

The same $n$ scales both the required source information and the available edge capacities. For this choice of $n$, the total communication rate on an undirected edge $uv$ is
\begin{equation}
f(uv)=\frac{b_{uv}+b_{vu}}{n}\le c(uv),
\label{eq:edge-rate}
\end{equation}
where the inequality follows from \eqref{eq:edge-capacity-budget}. Since $Y_{uv}$ is a binary string of length $b_{uv}$,
\[
H(Y_{uv})+H(Y_{vu})\le nf(uv).
\]

\noindent The \emph{network-coding capacity region} $\mathcal C_{\NC}(G,c,\cH)$ is the closure of the set of supported rate vectors \(r\).

For each session $h$, let $\mathcal P_h$ be the set of simple paths from $s_h$ to $t_h$. A \emph{fractional routing} of $r\in\Rge^k$ assigns nonnegative path rates $x_{h,P}$ satisfying
\begin{align*}
\sum_{P\in\mathcal P_h}x_{h,P}&=r_h &&\text{for every }h,\\
\sum_h\sum_{\substack{P\in\mathcal P_h\\e\in P}}x_{h,P}&\le c(e) &&\text{for every }e\in E.
\end{align*}

\noindent The \emph{fractional-routing capacity region} $\mathcal C_{\MCF}(G,c,\cH)$ is the set of rate vectors \(r\) admitting such a routing.

Since routing is a special case of network coding, $\mathcal C_{\MCF}(G,c,\cH)\subseteq\mathcal C_{\NC}(G,c,\cH)$. When the model $(G,c,\cH)$ is clear from the context, we write $\mathcal C_{\NC}$ and $\mathcal C_{\MCF}$ for short.

The \emph{demand graph} $D(\cH)$ has the terminals as vertices, and its edges are the distinct unordered pairs $\{s_h,t_h\}$ \cite[Sec.~III-A]{LiuQueLiLi2026}. Its \emph{matching number} $\nu(D(\cH))$ is the maximum number of edges with pairwise disjoint endpoints. Our main result is the following.
\begin{theorem}
\label{thm:intro-main}
If $\nu(D(\cH))\le2$, then
\begin{equation}
\mathcal C_{\NC}(G,c,\cH)=\mathcal C_{\MCF}(G,c,\cH).
\end{equation}
\end{theorem}

By the reduction in \cite{LiuQueLiLi2026} and the known inequalities for $K_2$, $K_{2,3}$, and $K_{3,3}$ \cite[Corollary~1]{Yin2018}, it remains to prove the distance-weighted communication bound on $\Gamma_{3,3}$. The proof of \cite[Theorem~8]{LiuQueLiLi2026} shows that it suffices to establish \cite[eq.~(7)]{LiuQueLiLi2026} with $Q=\Gamma_{3,3}$ for arbitrary sessions. In our notation, this bound is
\begin{equation}
\sum_{e\in E(\Gamma_{3,3})}f(e)
\ge \sum_h r_h d_{\Gamma_{3,3}}(s_h,t_h),
\label{eq:gamma-target}
\end{equation}
where $d_{\Gamma_{3,3}}(u,v)$ is the shortest-path distance between $u$ and $v$ when every edge has length one.

\section{The Directional Inequality on \texorpdfstring{$\Gamma_{3,3}$}{Gamma 3,3}}
\label{sec:gamma}
\subsection{The Graph and Its Edge Directions}
The graph $\Gamma_{3,3}$ used in Hirai's metric packing \cite{Hirai2010} is depicted in Fig.~\ref{fig:gamma-full}. It is obtained from $K_{3,3}$ by subdividing every edge once and joining all nine new vertices to one additional vertex, so it has 16 vertices and 27 edges.
Label the two vertex classes of $K_{3,3}$ by $a_0,a_1,a_2$ and $b_0,b_1,b_2$. For $i,j\in\{0,1,2\}$, let $x_{ij}$ be the vertex inserted on the edge $a_ib_j$, and let $o$ be the additional vertex. The edges of $\Gamma_{3,3}$ are exactly $a_ix_{ij}$, $b_jx_{ij}$, and $ox_{ij}$, for $i,j\in\{0,1,2\}$. We call the vertices $a_i$ and $b_j$ the \emph{outer terminals}, and we distinguish the \emph{$a$-terminals} from the \emph{$b$-terminals}. We call $x_{ij}$ the \emph{cell} in row $i$ and column $j$.

\begin{figure}[!tb]
\centering
\begin{tikzpicture}[v/.style={circle,draw,fill=white,inner sep=1.5pt,font=\scriptsize},scale=.88]
\foreach \i in {0,1,2} {
  \node[v] (a\i) at (-1.1,-1.3*\i) {$a_\i$};
  \node[v] (b\i) at (1.5*\i,1.05) {$b_\i$};
  \foreach \j in {0,1,2} {
    \node[v] (x\i\j) at (1.5*\j,-1.3*\i) {$x_{\i\j}$};
  }
}
\node[v] (o) at (4.25,-1.25) {$o$};
\foreach \i in {0,1,2} {
  \foreach \j in {0,1,2} {
    \draw (a\i) to[bend right=\j*15] (x\i\j);
    \draw (b\j) to[bend right=\i*15] (x\i\j);
    \draw[gray] (o) to[bend left=12] (x\i\j);
  }
}
\end{tikzpicture}
\caption{The undirected graph $\Gamma_{3,3}$. Each $x_{ij}$ is adjacent to exactly $a_i$, $b_j$, and $o$. Gray edges join $o$ to the nine cells. Crossings without a marked vertex are not vertices.}
\label{fig:gamma-full}
\end{figure}
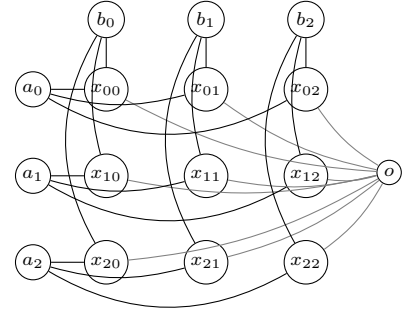

An \emph{arc} $(u,v)$ denotes the direction from $u$ to $v$ along an edge $uv$. For $0\le i<3$ and $0\le j<3$, we define the two arc sets
\begin{align}
E_+&=\{(b_j,x_{ij}),\ (o,x_{ij}),\ (x_{ij},a_i)\}_{i,j},\nonumber\\
E_-&=\{(v,u):(u,v)\in E_+\}.
\end{align}
The \emph{entropy costs} of the two arc sets are
\begin{equation}
C_+\triangleq\sum_{(u,v)\in E_+}H(Y_{uv}),\qquad C_-\triangleq\sum_{(u,v)\in E_-}H(Y_{uv}).
\end{equation}
Thus $C_+$ sums the transcript entropies in the directions $b_j\to x_{ij}$, $o\to x_{ij}$, and $x_{ij}\to a_i$ over all cells $x_{ij}$, while $C_-$ sums them in the reverse directions.
\begin{figure}[!tb]
\centering
\begin{tikzpicture}[v/.style={circle,draw,inner sep=2pt,font=\small},arr/.style={-{Latex[length=1.6mm]},thick}]
\node[v] (a) at (-1.5,0) {$a_i$};
\node[v] (x) at (0,0) {$x_{ij}$};
\node[v] (b) at (1.5,0) {$b_j$};
\node[v] (o) at (0,.95) {$o$};
\draw[arr] (b)--(x);\draw[arr] (o)--(x);\draw[arr] (x)--(a);
\end{tikzpicture}
\caption{One cell of $\Gamma_{3,3}$, with the directions in $E_+$. Reversing the arrows gives $E_-$. Rows share $a_i$, columns share $b_j$, and every cell is adjacent to $o$.}
\label{fig:cell}
\end{figure}
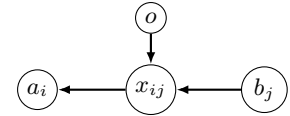

\subsection{Sessions between Outer Terminals}
We first consider sessions connecting two distinct $a$-terminals or two distinct $b$-terminals. Let $\cA$ (respectively, $\cB$) be the set of sessions connecting two distinct $a$-terminals (respectively, $b$-terminals), and define the message tuples
\begin{align*}
A&\triangleq(X_h:h\in\cA),\qquad B\triangleq(X_h:h\in\cB).
\end{align*}
If $\cA$ or $\cB$ is empty, the corresponding tuple is empty and has entropy zero. The directions of the sessions and the distributions of the messages are arbitrary. We write $R\triangleq H(A)+H(B)$.

\begin{theorem}[Directional inequality]
\label{thm:outer}
Every deterministic zero-error protocol on $\Gamma_{3,3}$ with session collection $\cA\cup\cB$ satisfies
\begin{equation}
C_+\ge2R,\qquad C_-\ge2R.
\label{eq:outer-bounds}
\end{equation}
\end{theorem}

Section~\ref{sec:proof} proves this inequality, and Section~\ref{sec:extension} extends it to arbitrary sessions.

\section{Proof of the Directional Inequality}
\label{sec:proof}
In this section, we prove Theorem~\ref{thm:outer}. We first prove the forward bound $C_+\ge2R$ and then the reverse bound $C_-\ge2R$.

For a vertex set $\mathcal Q$, let
\[
X_{\mathcal Q}^{\src}\triangleq(X_h:s_h\in\mathcal Q),\qquad X_{\mathcal Q}^{\dst}\triangleq(X_h:t_h\in\mathcal Q)
\]
denote the messages generated and requested in $\mathcal Q$, respectively. For the outer terminals, we write
\[
A_i^{\src}\triangleq X_{\{a_i\}}^{\src},\; A_i^{\dst}\triangleq X_{\{a_i\}}^{\dst},\; B_j^{\src}\triangleq X_{\{b_j\}}^{\src},\; B_j^{\dst}\triangleq X_{\{b_j\}}^{\dst}.
\]

Let $Y_{\mathcal Q,\le\tau}^{\mathrm{in}}$ denote the messages entering $\mathcal Q$ from outside through transmission $\tau$, and let $Y_{\mathcal Q,\le\tau}^{\mathrm{send}}$ denote the messages sent by vertices in $\mathcal Q$ through that transmission. Both parts use the following reconstruction lemma.

\begin{lemma}
\label{lem:region-simulation}
Every deterministic zero-error protocol with $N$ transmissions satisfies
\begin{equation}
H\!\left(Y_{\mathcal Q,\le\tau}^{\mathrm{send}}\mid X_{\mathcal Q}^{\src},Y_{\mathcal Q,\le\tau}^{\mathrm{in}}\right)=0
\quad (0\le\tau\le N),
\label{eq:region-prefix}
\end{equation}
and
\begin{equation}
H\!\left(X_{\mathcal Q}^{\dst}\mid X_{\mathcal Q}^{\src},Y_{\mathcal Q,\le N}^{\mathrm{in}}\right)=0.
\label{eq:region-decoding}
\end{equation}
\end{lemma}
Thus the local source messages and incoming communication determine all messages sent within or out of $\mathcal Q$, as well as the messages recovered there. The proof appears in Supplementary Section~\ref{supp:basic}.

\subsection{Forward Part}
We will prove that $C_+\ge2R$. By the definition of $E_+$,
\begin{align}
C_+&=\sum_{i,j}\bigl[H(Y_{b_jx_{ij}})+H(Y_{ox_{ij}})+H(Y_{x_{ij}a_i})\bigr]\nonumber\\
&\ge\sum_{i,j}\bigl[H(Y_{b_jx_{ij}})+H(Y_{ox_{ij}})\bigr]
+\sum_iH(\mathbf{Y}_i^a)\nonumber\\
&\ge\sum_{i,j}\bigl[H(Y_{b_jx_{ij}})+H(Y_{ox_{ij}})\bigr]
+\sum_iH(\mathbf{Y}_i^a\mid A_i^{\src}),
\label{eq:forward-expand}
\end{align}
where $\mathbf{Y}_i^a\triangleq(Y_{x_{ij}a_i})_{j=0}^2$ denotes the tuple of transcripts entering $a_i$. The first inequality follows from subadditivity, and the second from conditioning reducing entropy.

Next, we examine the first sum in \eqref{eq:forward-expand}. Let $\mathbf{Y}^a\triangleq(\mathbf{Y}_i^a)_{i=0}^2$ and define
\begin{align}
\kappa_{ij}&\triangleq H(Y_{b_jx_{ij}})+H(Y_{ox_{ij}})-H(Y_{b_jx_{ij}},Y_{ox_{ij}}\mid A,\mathbf{Y}^a)\nonumber\\
&=I(Y_{b_jx_{ij}},Y_{ox_{ij}};A,\mathbf{Y}^a)+I(Y_{b_jx_{ij}};Y_{ox_{ij}}).
\label{eq:kappa}
\end{align}
Continuing from \eqref{eq:forward-expand}, we have
{\footnotesize
\medmuskip=2mu
\thickmuskip=3mu
\begin{align}
C_+&\ge\sum_{i,j}\bigl[H(Y_{b_jx_{ij}})+H(Y_{ox_{ij}})\bigr]+\sum_iH(\mathbf{Y}_i^a\mid A_i^{\src})\nonumber\\
&=\sum_{i,j}\kappa_{ij}+\sum_{i,j}H(Y_{b_jx_{ij}},Y_{ox_{ij}}\mid A,\mathbf{Y}^a)+\sum_iH(\mathbf{Y}_i^a\mid A_i^{\src})\nonumber\\
&\ge\sum_{i,j}\kappa_{ij}+\sum_jH((Y_{b_jx_{ij}})_i,(Y_{ox_{ij}})_i\mid A,\mathbf{Y}^a)+\sum_iH(\mathbf{Y}_i^a\mid A_i^{\src}),
\label{eq:forward-decomposition}
\end{align}
}
where the equality follows from \eqref{eq:kappa}, and the last inequality follows from conditional subadditivity within each column.

Causality and the chain rule give
\begin{equation}
H(\mathbf{Y}_i^a\mid A_i^{\src})\le\sum_j\kappa_{ij}.
\label{eq:cell-charge}
\end{equation}
The proof appears in Supplementary Lemma~\ref{s-lem:cell-charge}.
Substituting \eqref{eq:cell-charge} into \eqref{eq:forward-decomposition} yields
{\thickmuskip=3mu
\begin{equation}
C_+\ge2\sum_iH(\mathbf{Y}_i^a\mid A_i^{\src})+\sum_jH((Y_{b_jx_{ij}})_i,(Y_{ox_{ij}})_i\mid A,\mathbf{Y}^a).
\label{eq:forward-two-rows}
\end{equation}
}

 Supplementary Lemma~\ref{s-lem:row-output} applies the decoding relation \eqref{eq:region-decoding} at each $a_i$ and obtains
\begin{equation}
\sum_iH(\mathbf{Y}_i^a\mid A_i^{\src})\ge H(A)+H(\mathbf{Y}^a\mid A).
\label{eq:row-output}
\end{equation}
Substituting \eqref{eq:row-output} into \eqref{eq:forward-two-rows}, we obtain
\begin{align}
C_+\ge{}&2H(A)+2H(\mathbf{Y}^a\mid A)\nonumber\\
&+\sum_jH((Y_{b_jx_{ij}})_i,(Y_{ox_{ij}})_i\mid A,\mathbf{Y}^a).
\label{eq:forward-before-columns}
\end{align}

For the last sum in \eqref{eq:forward-before-columns}, we have
\begin{equation}
\sum_jH((Y_{b_jx_{ij}})_i,(Y_{ox_{ij}})_i\mid A,\mathbf{Y}^a)\ge2H(B\mid A,\mathbf{Y}^a),
\label{eq:column-cover}
\end{equation}
 by Supplementary Lemmas~\ref{s-lem:column-equivalence} and~\ref{s-lem:column-cover}.
Therefore, by substituting \eqref{eq:column-cover} into \eqref{eq:forward-before-columns}, we obtain the forward bound:
\begin{align}
C_+&\ge2H(A)+2H(\mathbf{Y}^a\mid A)+2H(B\mid A,\mathbf{Y}^a)\nonumber\\
&=2H(A)+2H(B)=2R.
\label{eq:forward-final}
\end{align}
The equality follows because $A$ and $B$ are independent and $\mathbf{Y}^a$ is a deterministic function of $(A,B)$. Hence, the forward bound is proved.
\subsection{Reverse Part}
Now we will prove that $C_-\ge2R$. By the definition of $E_-$,
\begin{equation}
C_-=\sum_{i,j}\bigl[H(Y_{a_ix_{ij}})+H(Y_{x_{ij}b_j})+H(Y_{x_{ij}o})\bigr].
\label{eq:reverse-expand}
\end{equation}
We partition $E_-$ into the following five disjoint sets:
\begin{align*}
\cJ_0&=\{(x_{i0},b_0)\}_i,\qquad \cJ_2=\{(x_{i2},b_2)\}_i,\\
\cJ_1&=\{(x_{i0},o),(a_i,x_{i1}),(x_{i2},o)\}_i,\\
\cI_0&=\{(a_i,x_{i0}),(a_i,x_{i2}),(x_{i1},o)\}_i,\\
\cI_1&=\{(x_{i1},b_1)\}_i.
\end{align*}
Let $C_{\mathcal J}^-$ and $C_{\mathcal I}^-$ denote the entropy costs of the arcs in $\cJ_0\cup\cJ_1\cup\cJ_2$ and $\cI_0\cup\cI_1$, respectively. Then
\begin{equation}
C_-=\underbrace{\sum_{j=0}^2\sum_{e\in\cJ_j}H(Y_e)}_{C_{\mathcal J}^-}
+\underbrace{\sum_{l=0}^1\sum_{e\in\cI_l}H(Y_e)}_{C_{\mathcal I}^-}.
\label{eq:reverse-decomposition}
\end{equation}

\subsubsection{Bounding \texorpdfstring{$C_{\mathcal J}^-$}{C J minus}}
Let
\begin{equation}
Z_j\triangleq(B_j^{\src},(Y_e)_{e\in\cJ_j}),\qquad j=0,1,2.
\label{eq:fine-views}
\end{equation}
By Supplementary Lemma~\ref{s-lem:fine-complete},
\begin{equation}
H(Z_0,Z_1,Z_2)=H(A,B).
\label{eq:fine-complete}
\end{equation}
Define
\[
\mu\triangleq\sum_{j=0}^2H(Z_j)-H(Z_0,Z_1,Z_2)-H(B).
\]
Supplementary Lemma~\ref{s-lem:fine-bound} applies subadditivity to each tuple $Z_j$ and uses \eqref{eq:fine-complete} to obtain
\begin{equation}
C_{\mathcal J}^-\ge H(A)+H(B)+\mu.
\label{eq:fine-cost}
\end{equation}

Substituting \eqref{eq:fine-cost} into \eqref{eq:reverse-decomposition}, we obtain
\begin{equation}
C_-\ge H(A)+H(B)+\mu+C_{\mathcal I}^-.
\label{eq:reverse-after-fine}
\end{equation}

\subsubsection{Bounding \texorpdfstring{$C_{\mathcal I}^-$}{C I minus}}

Let
\[
W\triangleq\bigl(X_h:h\in\cB,\ \{s_h,t_h\}=\{b_0,b_2\}\bigr)
\]
denote the source messages in both directions between $b_0$ and $b_2$. Define
\begin{align}
\overline Z_0&\triangleq(B_0^{\src},B_2^{\src},(Y_e)_{e\in\cI_0}),\nonumber\\
\overline Z_1&\triangleq(B_1^{\src},(Y_e)_{e\in\cI_1}).
\end{align}
By Supplementary Lemma~\ref{s-lem:coarse-recovery}, we have
\begin{align}
H(\overline Z_0,\overline Z_1\mid W)&=H(A,B\mid W),\nonumber\\
I(\overline Z_0;\overline Z_1\mid W)&\ge H(B\mid W).
\label{eq:coarse-recovery}
\end{align}
Define
\begin{equation}
\begin{split}
P\triangleq{}&H(B_0^{\src})+H(B_2^{\src})+\sum_{e\in\cI_0}H(Y_e)-H(\overline Z_0).
\end{split}
\label{eq:P}
\end{equation}

Using \eqref{eq:coarse-recovery}, subadditivity, and $H(B)=H(W)+H(B\mid W)$, Supplementary Lemma~\ref{s-lem:coarse-bound} proves

\begin{equation}
C_{\mathcal I}^-\ge H(A)+H(B\mid W)+P.
\label{eq:coarse-cost}
\end{equation}
Together with \eqref{eq:reverse-after-fine}, we have
\begin{align}
C_-&\ge H(A)+H(B)+\mu+H(A)+H(B\mid W)+P\nonumber\\
&=2H(A)+2H(B)-H(W)+\mu+P.
\label{eq:reverse-partial}
\end{align}
To complete the reverse bound, we show that the remaining terms satisfy $\mu+P\ge H(W)$.

\subsubsection{Bounding \texorpdfstring{$\mu+P$}{mu + P}}
By Supplementary Lemma~\ref{s-lem:fine-bound},
\begin{equation}
\mu\ge I(W;Z_1).
\label{eq:mu-info}
\end{equation}
Consider the vertex sets
\begin{align*}
\mathcal Q_L&=\{b_2,o\}\cup\{x_{i2}:i=0,1,2\},\\
\mathcal Q_R&=\{b_0\}\cup\{x_{i0}:i=0,1,2\}.
\end{align*}
The edges with one endpoint in each set are $ox_{i0}$, $i=0,1,2$. Define
\[
\Pi\triangleq(Y_{ox_{i0}},Y_{x_{i0}o})_{i=0}^2.
\]
Since the transmission schedule is fixed, the messages in $\Pi$ are listed in the order in which they are transmitted in the protocol. Also define
\begin{align}
Z_L&\triangleq(B_2^{\src},(Y_{a_ix_{i2}})_i,(Y_{x_{i1}o})_i),\nonumber\\
Z_R&\triangleq(B_0^{\src},(Y_{a_ix_{i0}})_i).
\label{eq:LR-tapes}
\end{align}
Each tuple consists of the local source messages in the corresponding set and the transcripts entering that set from outside $\mathcal Q_L\cup\mathcal Q_R$. Up to reordering, we have
\begin{equation}
\overline Z_0=(Z_L,Z_R).
\label{eq:coarse-LR}
\end{equation}

Supplementary Lemma~\ref{s-lem:transcript-recovery} shows that $Z_1$ determines $\Pi$, i.e., $H(\Pi\mid Z_1)=0$. Thus, \eqref{eq:mu-info} and data processing give
\begin{equation}
\mu\ge I(W;Z_1)\ge I(W;\Pi).
\label{eq:mu-pi}
\end{equation}

Next, we bound $P$. By \eqref{eq:coarse-LR}, we have $H(\overline Z_0)=H(Z_L,Z_R)$. Substituting this equality into \eqref{eq:P}, we have
{\small
\setlength{\jot}{2pt}
\begin{align}
P&=H(B_0^{\src})+H(B_2^{\src})+\textstyle\sum_{e\in\cI_0}H(Y_e)\nonumber\\
&\quad-H(Z_L,Z_R)\nonumber\\
&=\Bigl[H(B_2^{\src})+\textstyle\sum_i\bigl(H(Y_{a_ix_{i2}})+H(Y_{x_{i1}o})\bigr)\Bigr]\nonumber\\
&\quad+\Bigl[H(B_0^{\src})+\textstyle\sum_iH(Y_{a_ix_{i0}})\Bigr]-H(Z_L,Z_R)\nonumber\\
&\ge H(Z_L)+H(Z_R)-H(Z_L,Z_R)=I(Z_L;Z_R).
\label{eq:P-tc}
\end{align}
}
The two bracketed sums correspond to the components of $Z_L$ and $Z_R$ in \eqref{eq:LR-tapes}, respectively, so the inequality follows from subadditivity. Combining \eqref{eq:P-tc} with \eqref{eq:mu-pi}, we obtain
\begin{equation}
\mu+P\ge I(W;\Pi)+I(Z_L;Z_R).
\label{eq:reverse-interaction-cost}
\end{equation}

Note that both $(Z_L,\Pi)$ and $(Z_R,\Pi)$ determine $W$: each contains the sources and incoming transcripts of its region. Every message in $\Pi$ is determined by the sending region's tuple and earlier messages. Supplementary Lemma~\ref{s-lem:interaction} therefore gives
\[
H(W\mid\Pi)\le I(Z_L;Z_R\mid\Pi)\le I(Z_L;Z_R).
\]
The first inequality holds because both $(Z_L,\Pi)$ and $(Z_R,\Pi)$ determine $W$. The second follows from the deterministic interaction property of $\Pi$. Thus,
\[
H(W)=I(W;\Pi)+H(W\mid\Pi)\le I(W;\Pi)+I(Z_L;Z_R).
\]
Together with \eqref{eq:reverse-interaction-cost}, we have
\[
\mu+P\ge H(W).
\]
 Substituting the above inequality into \eqref{eq:reverse-partial}, we obtain
\begin{align}
C_-&\ge2H(A)+2H(B)-H(W)+\mu+P\nonumber\\
&\ge2H(A)+2H(B)-H(W)+H(W)=2R.
\end{align}
This proves the reverse bound and completes the proof of Theorem~\ref{thm:outer}.

\section{Arbitrary Sessions and the Main Theorem}
\label{sec:extension}
In this section, we will extend the directional inequality to sessions with arbitrary endpoints on $\Gamma_{3,3}$ and obtain \eqref{eq:gamma-target}, completing the proof of Theorem~\ref{thm:intro-main}.

Recall that $d_{\Gamma_{3,3}}(u,v)$ is the shortest-path distance between $u$ and $v$ in $\Gamma_{3,3}$ when every edge has length one. Throughout this section, we write $d(u,v)$ for $d_{\Gamma_{3,3}}(u,v)$.
\begin{theorem}
\label{thm:all-session}
Every deterministic zero-error protocol on $\Gamma_{3,3}$ with independent source messages $X_h$ satisfies
\begin{equation}
\sum_{uv\in E(\Gamma_{3,3})}[H(Y_{uv})+H(Y_{vu})]\ge\sum_hH(X_h)d(s_h,t_h).
\label{eq:all-session}
\end{equation}
\end{theorem}
\begin{IEEEproof}
Consider a session $h$ with source $s_h$ and destination $t_h$. To apply Theorem~\ref{thm:outer}, we extend this session to one with source $\widehat s_h$ and destination $\widehat t_h$ satisfying either
\[
\widehat s_h=a_i,\quad \widehat t_h=a_k,
\quad i,k\in\{0,1,2\},\quad i\ne k,
\]
or
\[
\widehat s_h=b_j,\quad \widehat t_h=b_\ell,
\quad j,\ell\in\{0,1,2\},\quad j\ne\ell.
\]
Supplementary Lemma~\ref{lem:geodesic} guarantees that these endpoints can be chosen so that a shortest path of length four from $\widehat s_h$ to $\widehat t_h$ passes through $s_h$ and then $t_h$. Note that $s_h$ may coincide with $\widehat s_h$, and $t_h$ may coincide with $\widehat t_h$. Splitting the chosen length-four path at $s_h$ and $t_h$ gives three portions whose lengths sum to four. Since replacing any portion by a shorter path would contradict the shortestness of the chosen path, each portion is a shortest path between its endpoints. Thus,
\[
d(\widehat s_h,s_h)+d(s_h,t_h)+d(t_h,\widehat t_h)=4.
\]
For each original session $h$, we construct a corresponding new session with the same message $X_h$, initially held at $\widehat s_h$ and requested at $\widehat t_h$. We construct a protocol for these new sessions by extending the original protocol as follows. The joint distribution of the source messages remains unchanged. For each session $h$, we first forward $X_h$ from $\widehat s_h$ to $s_h$ along the chosen path. Once every original source $s_h$ has received its message $X_h$, we execute the original protocol. Each original destination $t_h$ then recovers $X_h$, which we forward along the chosen path to $\widehat t_h$. This defines a deterministic zero-error protocol for the new collection of sessions.

Let $Y'_{uv}$ denote the directed transcripts of the extended protocol, and define
{\small
\setlength{\jot}{2pt}
\begin{align*}
T&\triangleq\textstyle\sum_{uv\in E(\Gamma_{3,3})}[H(Y_{uv})+H(Y_{vu})],\\
T'&\triangleq\textstyle\sum_{uv\in E(\Gamma_{3,3})}[H(Y'_{uv})+H(Y'_{vu})].
\end{align*}
}
Each message $X_h$ is forwarded over $d(\widehat s_h,s_h)+d(t_h,\widehat t_h)=4-d(s_h,t_h)$ additional edges. Each extended transcript consists of the original transcript and the source messages forwarded on that directed edge. Therefore, by subadditivity, we have
\[
T'\le T+\sum_h[4-d(s_h,t_h)]H(X_h).
\]
For the extended protocol, define $C_\pm'\triangleq\sum_{(u,v)\in E_\pm}H(Y'_{uv})$.  Since the source and destination of each new session lie on the same outer side, it follows from Theorem~\ref{thm:outer} that
\[
T'=C_+'+C_-'\ge4\sum_hH(X_h).
\]
Combining the two inequalities, we obtain
\begin{align*}
T&\ge T'-\textstyle\sum_h[4-d(s_h,t_h)]H(X_h)\\
&\ge4\textstyle\sum_hH(X_h)-\textstyle\sum_h[4-d(s_h,t_h)]H(X_h)\\
&=\textstyle\sum_h d(s_h,t_h)H(X_h),
\end{align*}
which proves \eqref{eq:all-session}.
\end{IEEEproof}

Since the protocol decodes correctly for every source-message combination, we may take the source messages to be independent and uniform on their respective alphabets without changing the protocol or its communication cost. For a supported rate vector $r$ and the corresponding integer $n$, \eqref{eq:source-rate-budget} and \eqref{eq:edge-rate} give
\begin{align*}
H(X_h)&=\log_2|\mathcal X_h|\ge nr_h,\\
H(Y_{uv})+H(Y_{vu})&\le b_{uv}+b_{vu}=nf(uv).
\end{align*}
Hence, Theorem~\ref{thm:all-session} yields
\begin{align*}
n\textstyle\sum_{e\in E(\Gamma_{3,3})}f(e)
&\ge\textstyle\sum_{uv\in E(\Gamma_{3,3})}[H(Y_{uv})+H(Y_{vu})]\\
&\ge\textstyle\sum_hH(X_h)d(s_h,t_h)\\
&\ge n\textstyle\sum_h r_hd(s_h,t_h).
\end{align*}
Dividing the above inequality by $n$ on both sides, we have
\begin{equation}
\sum_{e\in E(\Gamma_{3,3})}f(e)\ge\sum_h r_hd(s_h,t_h).
\label{eq:gamma-rate}
\end{equation}

Inequality~\eqref{eq:gamma-rate} is precisely the desired bound stated in \eqref{eq:gamma-target}. Therefore, inequality~\eqref{eq:gamma-rate}, together with the reduction in \cite[proof of Theorem~8]{LiuQueLiLi2026}, completes the proof of Theorem~\ref{thm:intro-main}. A detailed metric-packing argument, including a direct proof of the transfer lemma, is provided in Supplementary Section~\ref{supp:main-reduction}.

\section{Conclusion}
\label{sec:conclusion}
We have established a distance-weighted communication inequality for arbitrary unicast sessions on $\Gamma_{3,3}$: the total communication rate is at least the sum of the session rates weighted by their source--destination shortest-path distances. The proof first establishes two directional entropy inequalities for sessions between designated outer terminals and then extends the resulting bound to sessions with arbitrary endpoints. Combined with the metric-packing reduction of Liu et al., this inequality proves that, for every undirected network and every edge-capacity assignment, network coding and fractional routing have the same capacity region whenever the demand graph has matching number at most two, i.e., no three demand pairs have pairwise disjoint endpoints.

\clearpage
\section*{Acknowledgments}
\textit{AI statement:} The authors used ChatGPT and Codex for assistance with brainstorming, checking the correctness of proofs, and language revision in the main text and supplementary material. The authors are responsible for the mathematical claims and the final manuscript.

\IEEEtriggeratref{8}

\clearpage

\setcounter{section}{0}
\setcounter{equation}{0}
\setcounter{lemma}{0}
\setcounter{proposition}{0}
\setcounter{table}{0}
\renewcommand{\thesection}{S\Roman{section}}
\renewcommand{\theequation}{S\arabic{equation}}
\renewcommand{\thelemma}{S\arabic{lemma}}
\renewcommand{\theproposition}{S\arabic{proposition}}
\renewcommand{\thetable}{S\Roman{table}}
\renewcommand{\theHsection}{supp.\arabic{section}}
\renewcommand{\theHequation}{supp.\arabic{equation}}
\renewcommand{\theHlemma}{supp.\arabic{lemma}}
\renewcommand{\theHproposition}{supp.\arabic{proposition}}
\renewcommand{\theHtable}{supp.\arabic{table}}

\begin{center}
{\large\bfseries Supplementary Material}\\[3pt]
{\normalsize Routing Optimality for Demand Graphs\\of Matching Number at Most Two}
\end{center}
We use the notation and the protocol model of the main paper, and all entropies are measured in bits. This supplement completes the proofs of the main paper. Section~\ref{supp:basic} proves the two reconstruction principles used by every entropy calculation below. Section~\ref{supp:forward} proves the forward bound $C_+\ge2R$, and Section~\ref{supp:reverse} proves the reverse bound $C_-\ge2R$, completing the proof of Theorem~\ref{thm:outer}. Section~\ref{supp:geodesics} verifies the geodesic table that extends the directional bounds to arbitrary sessions, and Section~\ref{supp:main-reduction} gives the metric-packing argument that completes the proof of Theorem~\ref{thm:intro-main}.

Recall that $\Gamma_{3,3}$ has vertices $a_i,b_j,x_{ij}$ and $o$, where $i,j\in\{0,1,2\}$, and edges $a_ix_{ij}$, $b_jx_{ij}$, and $ox_{ij}$. The vertices $a_i$ and $b_j$ are the outer terminals.

\section{Reconstruction Principles}
\label{supp:basic}
Both directional bounds are proved by reconstructing parts of the protocol from partial transcript data. We state and prove the two reconstruction principles here. The first says that a region of the network can reproduce all of its own activity once it knows its local source messages and everything entering it from outside. The second recovers source messages from the transcripts crossing a partition; it is used in the reverse bound to recover the $a$-sources from the $b$-sources and the crossing transcripts. Both arguments concern realizable values of the sources and the transcripts.

Recall from Section~\ref{sec:model} that $X_h$ is the source message of session $h\in\mathcal H$, generated at $s_h$ and requested at $t_h$. For a vertex set $\mathcal Q$, the notation of Section~\ref{sec:proof} groups these messages as
\[
X_{\mathcal Q}^{\src}=(X_h:s_h\in\mathcal Q),\qquad
X_{\mathcal Q}^{\dst}=(X_h:t_h\in\mathcal Q).
\]
The protocol has $N$ transmissions, indexed by $\tau=1,\ldots,N$. We write $Y_\tau$ for the message sent at transmission $\tau$, whereas $Y_{uv}$ denotes the complete transcript sent from $u$ to $v$. The tuple $Y_{\mathcal Q,\le\tau}^{\mathrm{in}}$ consists of messages entering $\mathcal Q$ from outside through transmission $\tau$. The tuple $Y_{\mathcal Q,\le\tau}^{\mathrm{send}}$ consists of all messages sent by vertices in $\mathcal Q$ through that transmission, including messages sent within $\mathcal Q$. Both tuples are listed in transmission order and are empty at $\tau=0$.

\begin{lemma}
\label{s-lem:region-simulation}
Let $\mathcal Q$ be a set of network vertices. Every deterministic zero-error protocol with $N$ transmissions satisfies
\begin{equation}
H\!\left(Y_{\mathcal Q,\le\tau}^{\mathrm{send}}\mid X_{\mathcal Q}^{\src},Y_{\mathcal Q,\le\tau}^{\mathrm{in}}\right)=0
\quad (0\le\tau\le N),
\label{s-eq:region-prefix}
\end{equation}
and
\begin{equation}
H\!\left(X_{\mathcal Q}^{\dst}\mid X_{\mathcal Q}^{\src},Y_{\mathcal Q,\le N}^{\mathrm{in}}\right)=0.
\label{s-eq:region-decoding}
\end{equation}
\end{lemma}

\begin{IEEEproof}
We prove Equation~\eqref{s-eq:region-prefix} by induction on $\tau$, writing $S_\tau\triangleq Y_{\mathcal Q,\le\tau}^{\mathrm{send}}$ and $I_\tau\triangleq Y_{\mathcal Q,\le\tau}^{\mathrm{in}}$ for short. For $\tau=0$ no transmission has occurred, so $S_0$ and $I_0$ are empty and $H(S_0\mid X_{\mathcal Q}^{\src},I_0)=0$. Assume the claim holds for $\tau-1$, and let $Y_\tau$ be the message sent at transmission $\tau$. If the sender of $Y_\tau$ lies outside $\mathcal Q$, then $S_\tau=S_{\tau-1}$ and $I_\tau$ refines $I_{\tau-1}$, so
\[
H(S_\tau\mid X_{\mathcal Q}^{\src},I_\tau)\le H(S_{\tau-1}\mid X_{\mathcal Q}^{\src},I_{\tau-1})=0.
\]
If the sender lies in $\mathcal Q$, then $S_\tau=(S_{\tau-1},Y_\tau)$. By the protocol model, $Y_\tau$ is a deterministic function of the sender's local source messages, which lie in $X_{\mathcal Q}^{\src}$, and of its previously received messages, which lie in $S_{\tau-1}\cup I_{\tau-1}$; hence $Y_\tau$ is a deterministic function of $(X_{\mathcal Q}^{\src},S_{\tau-1},I_{\tau-1})$. The chain rule and the induction hypothesis give
\[
\begin{aligned}
H(S_\tau\mid X_{\mathcal Q}^{\src},I_\tau)
&{=} H(S_{\tau-1}\mid X_{\mathcal Q}^{\src},I_\tau)\\
&\quad+H(Y_\tau\mid X_{\mathcal Q}^{\src},S_{\tau-1},I_\tau)=0.
\end{aligned}
\]
This proves Equation~\eqref{s-eq:region-prefix}. In particular, $H(S_N\mid X_{\mathcal Q}^{\src},I_N)=0$. The tuple $(X_{\mathcal Q}^{\src},S_N,I_N)$ determines the local sources and complete received history of every destination in $\mathcal Q$. Hence, zero-error decoding gives
\[
\begin{aligned}
H(X_{\mathcal Q}^{\dst}\mid X_{\mathcal Q}^{\src},I_N)
&=H(X_{\mathcal Q}^{\dst}\mid X_{\mathcal Q}^{\src},S_N,I_N)\\
&=0.
\end{aligned}
\]
The first equality uses that $S_N$ is determined by $(X_{\mathcal Q}^{\src},I_N)$. This proves Equation~\eqref{s-eq:region-decoding}.
\end{IEEEproof}

\vspace{2mm}
Recall that Theorem~\ref{thm:outer} concerns sessions in $\mathcal A\cup\mathcal B$, where $\mathcal A$ consists of sessions connecting two distinct $a$-terminals and $\mathcal B$ consists of sessions connecting two distinct $b$-terminals. The corresponding source-message tuples are
\[
A=(X_h:h\in\mathcal A),\qquad B=(X_h:h\in\mathcal B).
\]
Thus, $\mathcal A,\mathcal B$ specify sessions, while $A,B$ denote their source messages. Under the assumptions of that theorem, these two disjoint groups contain every source message, i.e., $\mathcal H=\mathcal A\cup\mathcal B$. All individual source messages are mutually independent.

\begin{lemma}
\label{s-lem:partition-recovery}
Consider a protocol as in Theorem~\ref{thm:outer}. Partition the vertices so that, for every $h\in\mathcal A$, the source $s_h$ and destination $t_h$ lie in different parts. Let $\mathsf T$ denote the complete transcripts in both directions on all edges whose endpoints lie in different parts. Then
\begin{equation}
H(A\mid B,\mathsf T)=0.
\label{s-eq:partition-recovery}
\end{equation}
\end{lemma}

\begin{IEEEproof}
We prove \eqref{s-eq:partition-recovery} by contradiction. Suppose that
\[
H(A\mid B,\mathsf T)>0.
\]
Since all variables are finite, there are two distinct values of $A$ compatible with the same values of $B$ and $\mathsf T$, i.e., there exist $a^{(1)}\ne a^{(2)}$, $b$, and $\mathsf t$ such that
\[
\Pr\{A=a^{(r)},B=b,\mathsf T=\mathsf t\}>0,
\qquad r=1,2.
\]
The tuples $A$ and $B$ together specify all source values. Recall from Section~\ref{sec:model} that the fixed transmission schedule and deterministic encoding rules then determine every transmission in order. We call the resulting run an \emph{execution}. Superscripts $(1)$ and $(2)$ denote source and transcript values in the executions specified by $(a^{(1)},b)$ and $(a^{(2)},b)$, respectively. In execution $r$, the source values are $(a^{(r)},b)$ and the crossing transcript is $\mathsf t$, i.e.,
\[
A^{(r)}=a^{(r)},\qquad
B^{(r)}=b,\qquad
\mathsf T^{(r)}=\mathsf t,
\qquad r=1,2.
\]
Since $A=(X_h:h\in\mathcal A)$, there is a session $h\in\mathcal A$ for which
\begin{equation}
X_h^{(1)}\ne X_h^{(2)}.
\label{s-eq:partition-distinct-source}
\end{equation}
Let $\mathcal Q$ be the part containing $s_h$. The partition assumption places its destination outside this part, i.e.,
\[
s_h\in\mathcal Q,\qquad t_h\notin\mathcal Q.
\]

Define a third source vector $x^{(3)}$ by
\[
x_g^{(3)}=
\begin{cases}
X_g^{(1)}, & s_g\in\mathcal Q,\\
X_g^{(2)}, & s_g\notin\mathcal Q,
\end{cases}
\qquad g\in\mathcal H.
\]
This specifies all source messages before the third execution begins. Each chosen source value has positive marginal probability, so independence gives
\[
\Pr\{X_g=x_g^{(3)}\text{ for all }g\in\mathcal H\}
=\prod_{g\in\mathcal H}\Pr\{X_g=x_g^{(3)}\}>0.
\]
Hence, these source values induce a valid execution, which we call execution $3$; its source values satisfy $X_g^{(3)}=x_g^{(3)}$ for every $g\in\mathcal H$. In particular,
\[
X_h^{(3)}=X_h^{(1)}\ne X_h^{(2)}.
\]

Let $Y_\tau^{(r)}$ be the message sent at transmission $\tau$ in execution $r$, and let $u_\tau$ be its sender. The fixed schedule makes this sender the same in all three executions. We prove by induction on $\tau$ that
\begin{equation}
Y_\tau^{(3)}=
\begin{cases}
Y_\tau^{(1)}, & u_\tau\in\mathcal Q,\\
Y_\tau^{(2)}, & u_\tau\notin\mathcal Q.
\end{cases}
\label{s-eq:partition-execution}
\end{equation}
At the first transmission, the sender has no previously received messages, so the claim follows from the chosen local sources and determinism. Suppose the claim holds before transmission $\tau$. We distinguish two cases.

\noindent\underline{Case 1: $u_\tau\in\mathcal Q$.} The sender's local sources agree with execution $1$, i.e.,
\[
\bigl(X_{\{u_\tau\}}^{\src}\bigr)^{(3)}
=\bigl(X_{\{u_\tau\}}^{\src}\bigr)^{(1)}.
\]
For every earlier message $Y_\sigma$ received by $u_\tau$ from a sender in $\mathcal Q$, the induction hypothesis gives $Y_\sigma^{(3)}=Y_\sigma^{(1)}$. For an earlier message received from outside $\mathcal Q$, we have
\[
Y_\sigma^{(3)}=Y_\sigma^{(2)}=Y_\sigma^{(1)}.
\]
The first equality follows from the induction hypothesis. The second follows because this transmission crosses the boundary of the part $\mathcal Q$ and is therefore recorded in the common transcript $\mathsf T^{(1)}=\mathsf T^{(2)}$. Thus, the sender's entire received history before transmission $\tau$ agrees with execution $1$, i.e.,
\[
\bigl(Y_{\{u_\tau\},\le\tau-1}^{\mathrm{in}}\bigr)^{(3)}
=\bigl(Y_{\{u_\tau\},\le\tau-1}^{\mathrm{in}}\bigr)^{(1)}.
\]
Together with equality of the local sources, determinism gives $Y_\tau^{(3)}=Y_\tau^{(1)}$.

\noindent\underline{Case 2: $u_\tau\notin\mathcal Q$.} The construction similarly gives
\[
\bigl(X_{\{u_\tau\}}^{\src}\bigr)^{(3)}
=\bigl(X_{\{u_\tau\}}^{\src}\bigr)^{(2)}.
\]
Earlier messages received from outside $\mathcal Q$ satisfy $Y_\sigma^{(3)}=Y_\sigma^{(2)}$ by induction. Earlier messages received from $\mathcal Q$ satisfy $Y_\sigma^{(3)}=Y_\sigma^{(1)}=Y_\sigma^{(2)}$ by induction and equality of the crossing transcripts. Hence, the received histories also agree:
\[
\bigl(Y_{\{u_\tau\},\le\tau-1}^{\mathrm{in}}\bigr)^{(3)}
=\bigl(Y_{\{u_\tau\},\le\tau-1}^{\mathrm{in}}\bigr)^{(2)}.
\]
Determinism then gives $Y_\tau^{(3)}=Y_\tau^{(2)}$, completing the induction.

Since $t_h\notin\mathcal Q$, its local sources agree with execution $2$. By \eqref{s-eq:partition-execution}, every message it receives from outside $\mathcal Q$ satisfies $Y_\sigma^{(3)}=Y_\sigma^{(2)}$. Every message it receives from $\mathcal Q$ satisfies $Y_\sigma^{(3)}=Y_\sigma^{(1)}=Y_\sigma^{(2)}$, where the last equality follows from the common crossing transcript. Thus, both its local sources and its complete received history agree with execution $2$, i.e.,
\[
\begin{aligned}
\bigl(X_{\{t_h\}}^{\src}\bigr)^{(3)}
&=\bigl(X_{\{t_h\}}^{\src}\bigr)^{(2)},\\
\bigl(Y_{\{t_h\},\le N}^{\mathrm{in}}\bigr)^{(3)}
&=\bigl(Y_{\{t_h\},\le N}^{\mathrm{in}}\bigr)^{(2)}.
\end{aligned}
\]
The deterministic decoder at $t_h$ must therefore produce the same output in both executions. Zero-error decoding implies
\[
X_h^{(1)}=X_h^{(3)}=X_h^{(2)},
\]
contradicting \eqref{s-eq:partition-distinct-source}. Thus, on the support of $(B,\mathsf T)$, each pair of values is compatible with only one value of $A$. Therefore, $A$ is a deterministic function of $(B,\mathsf T)$, and
\[
H(A\mid B,\mathsf T)=0.
\]
\end{IEEEproof}

\section{The Forward Bound}
\label{supp:forward}
We prove the forward inequality $C_+\ge2R$ of Theorem~\ref{thm:outer}. Recall that
\begin{align*}
C_+&=\sum_{i,j}\bigl[H(Y_{b_jx_{ij}})+H(Y_{ox_{ij}})+H(Y_{x_{ij}a_i})\bigr],\\
R&=H(A)+H(B).
\end{align*}
The transcripts entering the $a$-terminals are grouped as
\[
\mathbf Y_i^a=(Y_{x_{ij}a_i})_{j=0}^2,
\qquad \mathbf Y^a=(\mathbf Y_i^a)_{i=0}^2.
\]
Thus $\mathbf Y_i^a$ contains all transcripts received by $a_i$. By \eqref{eq:kappa},
\begin{align*}
\kappa_{ij}
&=I(Y_{b_jx_{ij}},Y_{ox_{ij}};A,\mathbf Y^a)\\
&\quad+I(Y_{b_jx_{ij}};Y_{ox_{ij}})\\
&\ge I(Y_{b_jx_{ij}},Y_{ox_{ij}};A,\mathbf Y^a).
\end{align*}
We first bound $H(\mathbf Y_i^a\mid A_i^{\src})$ in terms of $\kappa_{ij}$. We then use decoding at the $a$- and $b$-terminals to obtain the contributions $2H(A)$ and $2H(B)$.
\subsection{Bounds for the row transcripts}

Recall that the source messages generated and requested at $a_i$ are
\[
A_i^{\src}=(X_h:s_h=a_i),\qquad A_i^{\dst}=(X_h:t_h=a_i).
\]

\begin{lemma}
\label{s-lem:cell-charge}
For every row $i$,
\begin{equation}
H(\mathbf{Y}_i^{a}\mid A_i^{\src})\le\sum_{j=0}^2\kappa_{ij}.
\label{s-eq:cell-charge}
\end{equation}
\end{lemma}

\begin{IEEEproof}
Fix a row $i$. The records $Y_{\{a_i\},\le N}^{\mathrm{in}}$ and $\mathbf{Y}_i^a$ differ only by a fixed reordering, so either may be used in entropy expressions. For each $j$, let $\mathcal T_{ij}$ be the set of transmission indices at which $x_{ij}$ sends to $a_i$. These sets partition the transmissions that deliver messages to $a_i$ according to the cell that sends them.

\noindent\underline{Step 1: Reconstruction at a transmission.} We first show that the message sent from $x_{ij}$ to $a_i$ is determined by $A_i^{\src}$, the earlier messages received by $a_i$, and the transcripts $Y_{b_jx_{ij}}$ and $Y_{ox_{ij}}$. Fix $\tau\in\mathcal T_{ij}$ and take $\mathcal Q=\{a_i,x_{ij}\}$. Since all session endpoints are outer terminals, $x_{ij}$ has no local source messages. Hence, $X_{\mathcal Q}^{\src}=A_i^{\src}$. The messages entering $\mathcal Q$ come from $b_j$ or $o$ into $x_{ij}$, or from $x_{ik}$ into $a_i$ with $k\ne j$. The former are contained in $Y_{b_jx_{ij}}$ and $Y_{ox_{ij}}$. The latter, up to transmission $\tau$, are contained in $Y_{\{a_i\},\le\tau-1}^{\mathrm{in}}$, since transmission $\tau$ is from $x_{ij}$ to $a_i$ and hence internal to $\mathcal Q$. Therefore these given transcripts determine the incoming messages of $\mathcal Q$, i.e.,
\[
H\!\left(Y_{\mathcal Q,\le\tau}^{\mathrm{in}}\mid Y_{\{a_i\},\le\tau-1}^{\mathrm{in}},Y_{b_jx_{ij}},Y_{ox_{ij}}\right)=0.
\]
Lemma~\ref{s-lem:region-simulation} then gives
\begin{align}
&H(Y_{\{a_i\},\le\tau}^{\mathrm{in}}\mid A_i^{\src},Y_{\{a_i\},\le\tau-1}^{\mathrm{in}},Y_{b_jx_{ij}},Y_{ox_{ij}})\nonumber\\
&\le H(Y_{\mathcal Q,\le\tau}^{\mathrm{send}}\mid A_i^{\src},Y_{\{a_i\},\le\tau-1}^{\mathrm{in}},Y_{b_jx_{ij}},Y_{ox_{ij}})\nonumber\\
&\le H(Y_{\mathcal Q,\le\tau}^{\mathrm{send}}\mid X_{\mathcal Q}^{\src},Y_{\mathcal Q,\le\tau}^{\mathrm{in}})=0.
\label{eq:forward-message-causality}
\end{align}
The first inequality holds because the new message received by $a_i$ is sent within $\mathcal Q$, and the last equality is \eqref{s-eq:region-prefix}.

\noindent\underline{Step 2: Summing over transmissions.} The chain rule gives
\begin{align}
&H(\mathbf{Y}_i^a\mid A_i^{\src})\nonumber\\
&=\sum_{\tau=1}^N H(Y_{\{a_i\},\le\tau}^{\mathrm{in}}\mid A_i^{\src},Y_{\{a_i\},\le\tau-1}^{\mathrm{in}})\nonumber\\
&=\sum_j\sum_{\tau\in\mathcal T_{ij}}H(Y_{\{a_i\},\le\tau}^{\mathrm{in}}\mid A_i^{\src},Y_{\{a_i\},\le\tau-1}^{\mathrm{in}})\nonumber\\
&=\sum_j\sum_{\tau\in\mathcal T_{ij}}\!I(Y_{b_jx_{ij}},Y_{ox_{ij}};Y_{\{a_i\},\le\tau}^{\mathrm{in}}\mid A_i^{\src},Y_{\{a_i\},\le\tau-1}^{\mathrm{in}})\nonumber\\
&\le\sum_j\sum_{\tau=1}^N\!I(Y_{b_jx_{ij}},Y_{ox_{ij}};Y_{\{a_i\},\le\tau}^{\mathrm{in}}\mid A_i^{\src},Y_{\{a_i\},\le\tau-1}^{\mathrm{in}})\nonumber\\
&=\sum_j I(Y_{b_jx_{ij}},Y_{ox_{ij}};\mathbf{Y}_i^a\mid A_i^{\src})\nonumber\\
&\le\sum_j I(Y_{b_jx_{ij}},Y_{ox_{ij}};A,\mathbf{Y}^a)\nonumber\\
&\le\sum_j\kappa_{ij}.
\label{s-eq:cell-charge-chain}
\end{align}
The second equality omits the indices at which $a_i$ receives no message, and the third uses \eqref{eq:forward-message-causality}. The first inequality adds the nonnegative conditional mutual-information terms omitted from each fixed $j$.

For the fourth equality, fix $j$. Each received-message prefix contains the preceding prefix, so
\begin{align*}
&I(Y_{b_jx_{ij}},Y_{ox_{ij}};Y_{\{a_i\},\le\tau}^{\mathrm{in}}\mid A_i^{\src},Y_{\{a_i\},\le\tau-1}^{\mathrm{in}})\\
&=H(Y_{b_jx_{ij}},Y_{ox_{ij}}\mid A_i^{\src},Y_{\{a_i\},\le\tau-1}^{\mathrm{in}})\\
&\quad-H(Y_{b_jx_{ij}},Y_{ox_{ij}}\mid A_i^{\src},Y_{\{a_i\},\le\tau}^{\mathrm{in}}).
\end{align*}
Summing over $\tau=1,\ldots,N$ cancels consecutive conditional-entropy terms. The initial prefix is empty, leaving
\begin{align*}
&H(Y_{b_jx_{ij}},Y_{ox_{ij}}\mid A_i^{\src})\\
&\quad-H(Y_{b_jx_{ij}},Y_{ox_{ij}}\mid A_i^{\src},Y_{\{a_i\},\le N}^{\mathrm{in}})\\
&=I(Y_{b_jx_{ij}},Y_{ox_{ij}};Y_{\{a_i\},\le N}^{\mathrm{in}}\mid A_i^{\src})\\
&=I(Y_{b_jx_{ij}},Y_{ox_{ij}};\mathbf Y_i^a\mid A_i^{\src}).
\end{align*}
The last equality uses that $Y_{\{a_i\},\le N}^{\mathrm{in}}$ lists the received messages in transmission order, whereas $\mathbf Y_i^a$ groups them by sending cell. The fixed schedule determines the reordering between these two records. Summing over $j$ gives the fourth equality in \eqref{s-eq:cell-charge-chain}.

For the second inequality, the chain rule and the inclusion of $(A_i^{\src},\mathbf Y_i^a)$ in $(A,\mathbf Y^a)$ give, for each $j$,
\begin{align*}
&I(Y_{b_jx_{ij}},Y_{ox_{ij}};\mathbf Y_i^a\mid A_i^{\src})\\
&\quad\le I(Y_{b_jx_{ij}},Y_{ox_{ij}};A_i^{\src},\mathbf Y_i^a)\\
&\quad\le I(Y_{b_jx_{ij}},Y_{ox_{ij}};A,\mathbf Y^a).
\end{align*}
The last inequality in \eqref{s-eq:cell-charge-chain} follows from \eqref{eq:kappa}.
\end{IEEEproof}

\begin{lemma}
\label{s-lem:row-output}
The decoding requirements at the $a$-terminals imply
\begin{equation}
\sum_iH(\mathbf{Y}_i^{a}\mid A_i^{\src})\ge H(A)+H(\mathbf{Y}^{a}\mid A).
\label{s-eq:row-output}
\end{equation}
\end{lemma}

\begin{IEEEproof}
Apply \eqref{s-eq:region-decoding} to $\mathcal Q=\{a_i\}$. Its local source tuple is $A_i^{\src}$ and its incoming transcript tuple is $\mathbf Y_i^a$, so
\[
H(A_i^{\dst}\mid A_i^{\src},\mathbf Y_i^a)=0.
\]
A session has distinct endpoints, so $A_i^{\src}$ and $A_i^{\dst}$ contain disjoint sets of independent source messages. The chain rule therefore gives
\begin{align}
H(\mathbf Y_i^a\mid A_i^{\src})
&=H(A_i^{\dst},\mathbf Y_i^a\mid A_i^{\src})\nonumber\\
&=H(A_i^{\dst})+H(\mathbf Y_i^a\mid A_i^{\src},A_i^{\dst})\nonumber\\
&\ge H(A_i^{\dst})+H(\mathbf Y_i^a\mid A).
\label{s-eq:row-single}
\end{align}
The last inequality conditions on the remaining components of $A$. Each $A$-message occurs in exactly one destination group. Since the source messages are independent,
\[
\sum_iH(A_i^{\dst})=H(A).
\]
Summing \eqref{s-eq:row-single} over $i$ and applying conditional subadditivity gives
\begin{align*}
\sum_iH(\mathbf Y_i^a\mid A_i^{\src})
&\ge\sum_iH(A_i^{\dst})+\sum_iH(\mathbf Y_i^a\mid A)\\
&\ge H(A)+H(\mathbf Y^a\mid A).
\end{align*}
\end{IEEEproof}

\subsection{A bound for the column transcripts}
Recall that the source messages generated and requested at $b_j$ are
\[
B_j^{\src}=(X_h:s_h=b_j),\qquad B_j^{\dst}=(X_h:t_h=b_j).
\]
For each column $j$, we compare the tuples $((Y_{b_jx_{ij}})_i,(Y_{ox_{ij}})_i)$ and $(B_j^{\src},(Y_{ox_{ij}})_i)$ when $(A,\mathbf Y^a)$ is given. We first establish a recovery property using the outgoing transcripts of all three $b$-terminals.
\begin{lemma}
\label{s-lem:upload}
\[
H\bigl(B\mid A,(Y_{b_jx_{ij}})_{i,j}\bigr)=0.
\]
\end{lemma}

\begin{IEEEproof}
Let $\mathcal R=V(\Gamma_{3,3})\setminus\{b_0,b_1,b_2\}$. Its local source tuple is $A$, and its incoming transcripts are $(Y_{b_jx_{ij}})_{i,j}$. Lemma~\ref{s-lem:region-simulation} therefore gives
\begin{equation}
H\bigl((Y_{x_{ij}b_j})_{i,j}\mid A,(Y_{b_jx_{ij}})_{i,j}\bigr)=0.
\label{s-eq:upload-incoming-recovery}
\end{equation}
Consider the partition $\mathcal R,\{b_0\},\{b_1\},\{b_2\}$, whose crossing transcripts form the tuple
\[
\mathsf T=(Y_{b_jx_{ij}},Y_{x_{ij}b_j})_{i,j}.
\]
We keep the three $b$-terminals in separate parts so that the source and destination of every session in $\mathcal B$ lie in different parts, as required by Lemma~\ref{s-lem:partition-recovery}. Applying that lemma with the roles of $A$ and $B$ interchanged gives
\[
H(B\mid A,\mathsf T)=0.
\]
By \eqref{s-eq:upload-incoming-recovery}, the tuple $(A,(Y_{b_jx_{ij}})_{i,j})$ determines $\mathsf T$. Hence,
\[
\begin{aligned}
&H\bigl(B\mid A,(Y_{b_jx_{ij}})_{i,j}\bigr)\\
&\quad=H\bigl(B\mid A,(Y_{b_jx_{ij}})_{i,j},\mathsf T\bigr)\\
&\quad=H(B\mid A,\mathsf T)=0.
\end{aligned}
\]
\end{IEEEproof}

\begin{lemma}
\label{s-lem:column-equivalence}
For every column $j$,
\begin{equation}
H((Y_{b_jx_{ij}})_i,(Y_{ox_{ij}})_i\mid A,\mathbf{Y}^{a})=H(B_j^{\src},(Y_{ox_{ij}})_i\mid A,\mathbf{Y}^{a}).
\label{s-eq:column-equivalence}
\end{equation}
\end{lemma}

\begin{IEEEproof}
We prove that, given $(A,\mathbf Y^a)$, each of the two tuples in \eqref{s-eq:column-equivalence} determines the other. Since $(Y_{ox_{ij}})_i$ is contained in both tuples, it suffices to recover $(Y_{b_jx_{ij}})_i$ from $(B_j^{\src},(Y_{ox_{ij}})_i)$ and to recover $B_j^{\src}$ from $((Y_{b_jx_{ij}})_i,(Y_{ox_{ij}})_i)$, in both cases with $(A,\mathbf Y^a)$ given.

\noindent\underline{Recovering $(Y_{b_jx_{ij}})_i$.} Given $(A,\mathbf Y^a)$, apply Lemma~\ref{s-lem:region-simulation} to each $\{a_i\}$. This determines all transcripts $Y_{a_ix_{ij}}$. For the region $\{b_j\}\cup\{x_{ij}:i=0,1,2\}$, the local source tuple is $B_j^{\src}$, and the incoming transcripts are $(Y_{a_ix_{ij}},Y_{ox_{ij}})_i$. The same lemma therefore gives
\[
H\bigl((Y_{b_jx_{ij}})_i\mid A,\mathbf Y^a,B_j^{\src},(Y_{ox_{ij}})_i\bigr)=0.
\]
\noindent\underline{Recovering $B_j^{\src}$.} We prove this direction by contradiction. If
\[
H\bigl(B_j^{\src}\mid A,\mathbf Y^a,(Y_{b_jx_{ij}})_i,(Y_{ox_{ij}})_i\bigr)>0,
\]
there are two positive-probability source realizations with the same conditioning tuple but different $B_j^{\src}$. With superscripts denoting their execution values as in Lemma~\ref{s-lem:partition-recovery}, this means
\begin{align*}
(A,\mathbf Y^a)^{(1)}&=(A,\mathbf Y^a)^{(2)},\\
(Y_{b_jx_{ij}})_i^{(1)}&=(Y_{b_jx_{ij}})_i^{(2)},\\
(Y_{ox_{ij}})_i^{(1)}&=(Y_{ox_{ij}})_i^{(2)},\\
(B_j^{\src})^{(1)}&\ne(B_j^{\src})^{(2)}.
\end{align*}
Their common $(A,\mathbf Y^a)$ determines the same $(Y_{a_ix_{ij}})_i$. Each $x_{ij}$ consequently receives identical transcripts from its three neighbors in the two executions. Since $x_{ij}$ has no local source messages,
\[
(Y_{x_{ij}b_j})_i^{(1)}=(Y_{x_{ij}b_j})_i^{(2)}.
\]
Form execution $3$ by taking $B_j^{\src}$ from execution $1$ and all other sources from execution $2$. Independence ensures that these source values are admissible.
We show by induction that every transmission in execution $3$ agrees with execution $2$. Let $Y_\tau^{(r)}$ be the message at transmission $\tau$ in execution $r$, and let $u_\tau$ be its sender. Suppose that $Y_\sigma^{(3)}=Y_\sigma^{(2)}$ for all $\sigma<\tau$; at $\tau=1$, there are no earlier messages.

If $u_\tau\ne b_j$, the sender's local sources and earlier received messages agree with execution $2$, so determinism gives $Y_\tau^{(3)}=Y_\tau^{(2)}$. If $u_\tau=b_j$, its local sources agree with execution $1$. For each earlier message received by $b_j$, induction and equality of the incoming transcripts in executions $1$ and $2$ give
\[
Y_\sigma^{(3)}=Y_\sigma^{(2)}=Y_\sigma^{(1)}.
\]
Thus $b_j$ sends the same message as in execution $1$. Since its outgoing transcripts also agree in executions $1$ and $2$,
\[
Y_\tau^{(3)}=Y_\tau^{(1)}=Y_\tau^{(2)}.
\]
This completes the induction.

Thus,
\begin{align}
A^{(3)}&=A^{(2)},\nonumber\\
(Y_{b_kx_{ik}})_{i,k}^{(3)}&=(Y_{b_kx_{ik}})_{i,k}^{(2)},\nonumber\\
(B_j^{\src})^{(3)}&=(B_j^{\src})^{(1)}\ne(B_j^{\src})^{(2)}.
\label{s-eq:column-executions}
\end{align}
The first two equalities in \eqref{s-eq:column-executions} imply $B^{(3)}=B^{(2)}$ by Lemma~\ref{s-lem:upload}. In particular, $(B_j^{\src})^{(3)}=(B_j^{\src})^{(2)}$, contradicting the inequality in \eqref{s-eq:column-executions}. Hence,
\[
H\bigl(B_j^{\src}\mid A,\mathbf Y^a,(Y_{b_jx_{ij}})_i,(Y_{ox_{ij}})_i\bigr)=0.
\]
The two zero conditional entropies show that, given $(A,\mathbf Y^a)$, either tuple in the statement determines the other. Applying the chain rule in both orders proves \eqref{s-eq:column-equivalence}.
\end{IEEEproof}

For a tuple $X=(X_1,\ldots,X_n)$ and $S\subseteq\{1,\ldots,n\}$, write $X_S=(X_i:i\in S)$.

\begin{lemma}
\label{s-lem:shearer}
Let $X=(X_1,\ldots,X_n)$ and $S_1,\ldots,S_q\subseteq\{1,\ldots,n\}$ satisfy
\[
|\{\alpha:i\in S_\alpha\}|\ge k
\qquad(i=1,\ldots,n).
\]
Then, for every random variable $Z$,
\begin{equation}
\sum_{\alpha=1}^q H(X_{S_\alpha}\mid Z)\ge kH(X\mid Z).
\label{s-eq:shearer}
\end{equation}
\end{lemma}

\begin{IEEEproof}
Order the coordinates as $1,\ldots,n$. For each $S_\alpha$, the chain rule and conditioning reducing entropy give
\begin{align*}
H(X_{S_\alpha}\mid Z)
&=\sum_{i\in S_\alpha}H(X_i\mid Z,(X_j)_{j\in S_\alpha,\,j<i})\\
&\ge\sum_{i\in S_\alpha}H(X_i\mid Z,X_1,\ldots,X_{i-1}).
\end{align*}
Each coordinate belongs to at least $k$ sets $S_\alpha$, and all terms in the last sum are nonnegative. Therefore,
\begin{align*}
\sum_\alpha H(X_{S_\alpha}\mid Z)
&\ge k\sum_{i=1}^nH(X_i\mid Z,X_1,\ldots,X_{i-1})\\
&=kH(X\mid Z).
\end{align*}
\end{IEEEproof}

\begin{lemma}
\label{s-lem:column-cover}
The decoding requirements at the $b$-terminals imply
\begin{equation}
\sum_{j=0}^2H(B_j^{\src},(Y_{ox_{ij}})_i\mid A,\mathbf{Y}^{a})\ge2H(B\mid A,\mathbf{Y}^{a}).
\label{s-eq:column-cover}
\end{equation}
\end{lemma}

\begin{IEEEproof}
The reconstruction in Lemma~\ref{s-lem:column-equivalence} gives
\begin{equation}
H\bigl((Y_{x_{ij}b_j})_i\mid A,\mathbf Y^a,B_j^{\src},(Y_{ox_{ij}})_i\bigr)=0.
\label{s-eq:column-incoming-recovery}
\end{equation}
By \eqref{s-eq:region-decoding} with $\mathcal Q=\{b_j\}$,
\begin{equation}
H\bigl(B_j^{\dst}\mid B_j^{\src},(Y_{x_{ij}b_j})_i\bigr)=0.
\label{s-eq:column-terminal-decoding}
\end{equation}
Consequently,
\begin{equation}
\begin{aligned}
&H\bigl(B_j^{\dst}\mid A,\mathbf Y^a,B_j^{\src},(Y_{ox_{ij}})_i\bigr)\\
&\le H\bigl(B_j^{\dst},(Y_{x_{ij}b_j})_i\mid A,\mathbf Y^a,B_j^{\src},(Y_{ox_{ij}})_i\bigr)\\
&=H\bigl((Y_{x_{ij}b_j})_i\mid A,\mathbf Y^a,B_j^{\src},(Y_{ox_{ij}})_i\bigr)\\
&\quad+H\bigl(B_j^{\dst}\mid A,\mathbf Y^a,B_j^{\src},(Y_{ox_{ij}})_i,(Y_{x_{ij}b_j})_i\bigr)\\
&=H\bigl(B_j^{\dst}\mid A,\mathbf Y^a,B_j^{\src},(Y_{ox_{ij}})_i,(Y_{x_{ij}b_j})_i\bigr)\\
&\le H\bigl(B_j^{\dst}\mid B_j^{\src},(Y_{x_{ij}b_j})_i\bigr)=0.
\end{aligned}
\label{s-eq:column-destination-recovery}
\end{equation}
The first inequality adds the incoming transcripts to the tuple. The first equality is the chain rule, and the second uses \eqref{s-eq:column-incoming-recovery}. The last inequality follows from conditioning reducing entropy, and the final equality is \eqref{s-eq:column-terminal-decoding}.
The chain rule gives
\begin{align}
&H(B_j^{\src},(Y_{ox_{ij}})_i\mid A,\mathbf Y^a)\nonumber\\
&=H(B_j^{\src},B_j^{\dst},(Y_{ox_{ij}})_i\mid A,\mathbf Y^a)\nonumber\\
&\quad-H(B_j^{\dst}\mid A,\mathbf Y^a,B_j^{\src},(Y_{ox_{ij}})_i)\nonumber\\
&=H(B_j^{\src},B_j^{\dst},(Y_{ox_{ij}})_i\mid A,\mathbf Y^a)\nonumber\\
&=H(B_j^{\src},B_j^{\dst}\mid A,\mathbf Y^a)\nonumber\\
&\quad+H((Y_{ox_{ij}})_i\mid A,\mathbf Y^a,B_j^{\src},B_j^{\dst})\nonumber\\
&\ge H(B_j^{\src},B_j^{\dst}\mid A,\mathbf Y^a).
\label{s-eq:column-pair}
\end{align}
The first and third equalities use the chain rule. The second equality uses \eqref{s-eq:column-destination-recovery}, and the inequality follows from nonnegativity of the last conditional-entropy term.
Every $X_h$ with $h\in\mathcal B$ occurs in exactly two tuples $(B_j^{\src},B_j^{\dst})$, once at its source and once at its distinct destination. Apply Lemma~\ref{s-lem:shearer} to $B$, with conditioning tuple $(A,\mathbf Y^a)$ and $k=2$. Summing \eqref{s-eq:column-pair} then gives
\begin{align*}
\sum_jH(B_j^{\src},(Y_{ox_{ij}})_i\mid A,\mathbf Y^a)
&\ge\sum_jH(B_j^{\src},B_j^{\dst}\mid A,\mathbf Y^a)\\
&\ge2H(B\mid A,\mathbf Y^a).
\end{align*}
\end{IEEEproof}

\subsection{Completing the forward bound}

\begin{proposition}[Forward directional bound]
\label{s-prop:forward}
Every protocol in Theorem~\ref{thm:outer} satisfies $C_+\ge2R$.
\end{proposition}

\begin{IEEEproof}
Expanding $C_+$ using \eqref{eq:kappa} and applying conditional subadditivity gives
\begin{align}
C_+\ge{}&\sum_{i,j}\kappa_{ij}\nonumber\\
&+\sum_jH((Y_{b_jx_{ij}})_i,(Y_{ox_{ij}})_i\mid A,\mathbf{Y}^a)\nonumber\\
&+\sum_iH(\mathbf{Y}_i^a\mid A_i^{\src}).
\label{s-eq:forward-decomposition}
\end{align}
Summing Lemma~\ref{s-lem:cell-charge} over $i$ gives
\[
\sum_{i,j}\kappa_{ij}\ge\sum_iH(\mathbf Y_i^a\mid A_i^{\src}).
\]
Substituting into \eqref{s-eq:forward-decomposition} and applying Lemma~\ref{s-lem:row-output} gives
\begin{align}
C_+\ge{}&2H(A)+2H(\mathbf{Y}^a\mid A)\nonumber\\
&+\sum_jH((Y_{b_jx_{ij}})_i,(Y_{ox_{ij}})_i\mid A,\mathbf{Y}^a).
\label{s-eq:forward-before-columns}
\end{align}
For the remaining sum, Lemmas~\ref{s-lem:column-equivalence} and \ref{s-lem:column-cover} give
\begin{align*}
&\sum_jH((Y_{b_jx_{ij}})_i,(Y_{ox_{ij}})_i\mid A,\mathbf Y^a)\\
&=\sum_jH(B_j^{\src},(Y_{ox_{ij}})_i\mid A,\mathbf Y^a)\\
&\ge2H(B\mid A,\mathbf Y^a).
\end{align*}
The equality uses Lemma~\ref{s-lem:column-equivalence}, and the inequality uses Lemma~\ref{s-lem:column-cover}. Substituting into \eqref{s-eq:forward-before-columns} gives
\begin{align*}
C_+&\ge2H(A)+2H(\mathbf{Y}^a\mid A)+2H(B\mid A,\mathbf{Y}^a)\\
&=2H(A)+2H(B)=2R.
\end{align*}
For the equality, the chain rule gives
\begin{align*}
H(\mathbf Y^a\mid A)+H(B\mid A,\mathbf Y^a)
&=H(B,\mathbf Y^a\mid A)\\
&=H(B\mid A)=H(B).
\end{align*}
The second equality uses that $\mathbf Y^a$ is a deterministic function of $(A,B)$, and the last uses independence of $A$ and $B$. This proves the proposition.
\end{IEEEproof}

\section{The Reverse Bound}
\label{supp:reverse}
We prove the reverse inequality $C_-\ge2R$. Recall that
\[
C_-=\sum_{i,j}\bigl[H(Y_{a_ix_{ij}})+H(Y_{x_{ij}b_j})+H(Y_{x_{ij}o})\bigr].
\]
The reverse arcs are partitioned into
\begin{align*}
\mathcal J_0&=\{(x_{i0},b_0)\}_i,\qquad
\mathcal J_2=\{(x_{i2},b_2)\}_i,\\
\mathcal J_1&=\{(x_{i0},o),(a_i,x_{i1}),(x_{i2},o)\}_i,\\
\mathcal I_0&=\{(a_i,x_{i0}),(a_i,x_{i2}),(x_{i1},o)\}_i,\\
\mathcal I_1&=\{(x_{i1},b_1)\}_i,
\end{align*}
where $i=0,1,2$. For an arc $e=(u,v)$, write $Y_e=Y_{uv}$. Thus,
\begin{align*}
C_{\mathcal J}^-&=\sum_{j=0}^2\sum_{e\in\mathcal J_j}H(Y_e),\\
C_{\mathcal I}^-&=\sum_{l=0}^1\sum_{e\in\mathcal I_l}H(Y_e),\\
C_-&=C_{\mathcal J}^-+C_{\mathcal I}^-.
\end{align*}
We bound $C_{\mathcal J}^-$ and $C_{\mathcal I}^-$ separately and then combine the bounds to prove $C_-\ge2R$.

\subsection{Bounding \texorpdfstring{$C_{\mathcal J}^-$}{C J minus}}
\label{supp:fine}
Partition the vertices into
\begin{align}
U_0&=\{b_0\},\quad U_2=\{b_2\},\nonumber\\
U_1&=\{b_1,o\}\cup\{x_{i1}:0\le i<3\},\nonumber\\
V_i&=\{a_i,x_{i0},x_{i2}\}.
\label{eq:fine-partition}
\end{align}
Recall that
\[
Z_j=(B_j^{\src},(Y_e)_{e\in\mathcal J_j}),\qquad j=0,1,2.
\]

The only source messages generated in $U_j$ are those generated at $b_j$, and the arcs entering $U_j$ are exactly those in $\mathcal J_j$. Hence, up to a fixed reordering,
\[
X_{U_j}^{\src}=B_j^{\src},
\qquad
Y_{U_j,\le N}^{\mathrm{in}}=(Y_e)_{e\in\mathcal J_j}.
\]
Substituting these identities into the definition of $Z_j$ gives
\[
Z_j=(B_j^{\src},(Y_e)_{e\in\mathcal J_j})
=(X_{U_j}^{\src},Y_{U_j,\le N}^{\mathrm{in}}).
\]
Each edge whose endpoints lie in different parts has one endpoint in some $V_i$ and one in some $U_j$. Its arc directed into $U_j$ belongs to $\mathcal J_j$.
\begin{lemma}
\label{s-lem:fine-complete}
Together, $Z_0,Z_1,Z_2$ determine all sources:
\begin{equation}
H(Z_0,Z_1,Z_2)=H(A,B).
\label{s-eq:fine-complete}
\end{equation}
\end{lemma}

\begin{IEEEproof}
The tuples $B_j^{\src}$ contained in $Z_j$ together form $B$, so
\begin{equation}
H(B\mid Z_0,Z_1,Z_2)=0.
\label{s-eq:fine-B-recovery}
\end{equation}
For each $U_j$, apply \eqref{s-eq:region-prefix} at $\tau=N$:
\begin{equation}
\begin{aligned}
H(Y_{U_j,\le N}^{\mathrm{send}}\mid Z_j)
&=H(Y_{U_j,\le N}^{\mathrm{send}}\mid B_j^{\src},(Y_e)_{e\in\cJ_j})\\
&=0.
\end{aligned}
\label{s-eq:fine-outgoing-recovery}
\end{equation}
We now apply Lemma~\ref{s-lem:partition-recovery} to the six parts
\[
U_0,U_1,U_2,V_0,V_1,V_2
\]
defined in \eqref{eq:fine-partition}. Let $\mathsf T$ denote the complete bidirectional transcripts on edges joining different parts of this partition.

Every such edge joins some $V_i$ to some $U_j$. Its transcript directed into $U_j$ is contained in $Z_j$. Its reverse transcript, sent from $U_j$ to $V_i$, is a component of $Y_{U_j,\le N}^{\mathrm{send}}$ and is therefore determined by $Z_j$ by \eqref{s-eq:fine-outgoing-recovery}. Hence,
\begin{equation}
H(\mathsf T\mid Z_0,Z_1,Z_2)=0.
\label{s-eq:fine-T-recovery}
\end{equation}
The source messages to be recovered are $A$, with the remaining source messages $B$ given. Every session $h\in\mathcal A$ connects distinct terminals $a_i$ and $a_k$, which lie in distinct parts $V_i$ and $V_k$, respectively. Thus the partition satisfies the hypothesis of Lemma~\ref{s-lem:partition-recovery}, giving
\begin{equation}
H(A\mid B,\mathsf T)=0.
\label{s-eq:fine-A-recovery}
\end{equation}
Using the two determined tuples $B$ and $\mathsf T$, we obtain
\begin{align*}
H(A,B\mid Z_0,Z_1,Z_2)
&=H(A\mid B,Z_0,Z_1,Z_2)\\
&=H(A\mid B,\mathsf T,Z_0,Z_1,Z_2)\\
&\le H(A\mid B,\mathsf T)=0.
\end{align*}
The first equality follows from the chain rule and $H(B\mid Z_0,Z_1,Z_2)=0$ in \eqref{s-eq:fine-B-recovery}. The second equality uses \eqref{s-eq:fine-T-recovery}: since $(Z_0,Z_1,Z_2)$ determines $\mathsf T$, adding $\mathsf T$ to the conditioning does not change the entropy. The inequality follows because removing $(Z_0,Z_1,Z_2)$ from the conditioning cannot decrease entropy. The final equality is \eqref{s-eq:fine-A-recovery}.
Conversely, the sources determine all transcripts, so $H(Z_0,Z_1,Z_2\mid A,B)=0$. The chain rule now gives
\[
H(Z_0,Z_1,Z_2)=H(A,B,Z_0,Z_1,Z_2)=H(A,B).
\]
\end{IEEEproof}

Recall that
\[
\mu=\sum_{j=0}^2H(Z_j)-H(Z_0,Z_1,Z_2)-H(B).
\]

Recall also that $W$ consists of all source messages whose sessions connect $b_0$ and $b_2$, in either direction:
\[
W=(X_h:h\in\mathcal B,\ \{s_h,t_h\}=\{b_0,b_2\}).
\]

\begin{lemma}
\label{s-lem:fine-bound}
We have
\begin{align}
C_{\mathcal J}^-&\ge H(A)+H(B)+\mu,\label{s-eq:fine-cost}\\
\mu&\ge I(W;Z_1)\ge0.\label{s-eq:mu-info}
\end{align}
\end{lemma}

\begin{IEEEproof}
Since $Z_j=(B_j^{\src},(Y_e)_{e\in\mathcal J_j})$, subadditivity and independence of the source messages give
\begin{align}
\sum_jH(Z_j)
&\le\sum_jH(B_j^{\src})+\sum_j\sum_{e\in\mathcal J_j}H(Y_e)\nonumber\\
&=H(B)+C_{\mathcal J}^-.
\label{s-eq:fine-subadd}
\end{align}
Rearranging \eqref{s-eq:fine-subadd}, we obtain
\begin{align*}
C_{\mathcal J}^-&\ge\sum_jH(Z_j)-H(B)\\
&=H(Z_0,Z_1,Z_2)+\mu\\
&=H(A)+H(B)+\mu.
\end{align*}
The first equality uses the definition of $\mu$. The second uses \eqref{s-eq:fine-complete} and independence of $A$ and $B$.
To bound $\mu$, define the tuples
\[
K_{jk}\triangleq(X_h:h\in\mathcal B,\ \{s_h,t_h\}=\{b_j,b_k\})
\]
for $0\le j<k\le2$. In particular, $K_{02}=W$.
Thus $B=(K_{01},K_{12},W)$ up to reordering, and the three tuples are independent. Expanding the definition of $\mu$ gives
\begin{align}
\mu+H(B)
&=\sum_jH(Z_j)-H(Z_0,Z_1,Z_2)\nonumber\\
&=\bigl[H(Z_1)+H(Z_2)-H(Z_1,Z_2)\bigr]\nonumber\\
&\quad+\bigl[H(Z_0)+H(Z_1,Z_2)-H(Z_0,Z_1,Z_2)\bigr]\nonumber\\
&=I(Z_1;Z_2)+I(Z_0;Z_1,Z_2).
\label{s-eq:TC-chain}
\end{align}
Each $Z_j$ contains the local source messages and incoming transcripts of $U_j$. Applying \eqref{s-eq:region-decoding} to $U_j$ gives $H(B_j^{\dst}\mid Z_j)=0$. Since $B_j^{\src}$ is already contained in $Z_j$, the chain rule gives
\[
\begin{aligned}
H(B_j^{\src},B_j^{\dst}\mid Z_j)
&=H(B_j^{\src}\mid Z_j)\\
&\quad+H(B_j^{\dst}\mid B_j^{\src},Z_j)=0.
\end{aligned}
\]
Thus $Z_j$ determines all source messages generated or requested at $b_j$. Grouping these messages by their other endpoint gives, up to reordering,
\begin{align*}
(B_0^{\src},B_0^{\dst})&=(K_{01},W),\\
(B_1^{\src},B_1^{\dst})&=(K_{01},K_{12}),\\
(B_2^{\src},B_2^{\dst})&=(K_{12},W).
\end{align*}
Consequently,
\begin{equation}
\begin{aligned}
H(K_{01},W\mid Z_0)&=H(B_0^{\src},B_0^{\dst}\mid Z_0)=0,\\
H(K_{01},K_{12}\mid Z_1)&=H(B_1^{\src},B_1^{\dst}\mid Z_1)=0,\\
H(K_{12},W\mid Z_2)&=H(B_2^{\src},B_2^{\dst}\mid Z_2)=0.
\end{aligned}
\label{s-eq:pair-source-recovery}
\end{equation}
In particular, $Z_1$ determines $K_{01}$ and $Z_2$ determines $W$. The chain rule and conditioning reducing entropy therefore yield
\begin{equation}
\begin{aligned}
H(K_{01},W\mid Z_1,Z_2)
&=H(K_{01}\mid Z_1,Z_2)\\
&\quad+H(W\mid K_{01},Z_1,Z_2)\\
&\le H(K_{01}\mid Z_1)+H(W\mid Z_2)\\
&\le H(K_{01},K_{12}\mid Z_1)\\
&\quad+H(K_{12},W\mid Z_2)\\
&=H(B_1^{\src},B_1^{\dst}\mid Z_1)\\
&\quad+H(B_2^{\src},B_2^{\dst}\mid Z_2)\\
&=0,
\end{aligned}
\label{s-eq:joint-pair-recovery}
\end{equation}
where the second inequality adds $K_{12}$ to each entropy argument, and the last two equalities use \eqref{s-eq:pair-source-recovery}.

We now bound the two mutual informations in \eqref{s-eq:TC-chain}. For the first,
\begin{align}
I(Z_1;Z_2)&\ge I(Z_1;K_{12},W)\nonumber\\
&=I(Z_1;W)+I(Z_1;K_{12}\mid W)\nonumber\\
&=I(Z_1;W)+H(K_{12}\mid W)\nonumber\\
&\quad-H(K_{12}\mid Z_1,W)\nonumber\\
&=I(Z_1;W)+H(K_{12}\mid W)\nonumber\\
&=I(Z_1;W)+H(K_{12}).
\label{s-eq:first-mi}
\end{align}
The inequality uses data processing and $H(K_{12},W\mid Z_2)=0$ from \eqref{s-eq:pair-source-recovery}. The first equality is the chain rule, and the second expands conditional mutual information. For the third equality, conditioning reducing entropy and \eqref{s-eq:pair-source-recovery} give
\[
\begin{aligned}
0\le H(K_{12}\mid Z_1,W)
&\le H(K_{12}\mid Z_1)\\
&\le H(K_{01},K_{12}\mid Z_1)=0.
\end{aligned}
\]
The last equality in \eqref{s-eq:first-mi} uses independence of $K_{12}$ and $W$.

For the second mutual information,
\begin{align}
I(Z_0;Z_1,Z_2)&\ge I(K_{01},W;Z_1,Z_2)\nonumber\\
&=H(K_{01},W)-H(K_{01},W\mid Z_1,Z_2)\nonumber\\
&=H(K_{01},W)\nonumber\\
&=H(K_{01})+H(W).
\label{s-eq:second-mi}
\end{align}
The inequality uses data processing and $H(K_{01},W\mid Z_0)=0$ from \eqref{s-eq:pair-source-recovery}. The first equality expands mutual information, the second uses \eqref{s-eq:joint-pair-recovery}, and the last uses independence of $K_{01}$ and $W$.
Substitution into \eqref{s-eq:TC-chain} gives
\begin{align*}
\mu+H(B)&\ge I(W;Z_1)+H(K_{12})+H(K_{01})+H(W)\\
&=I(W;Z_1)+H(B).
\end{align*}
Subtracting $H(B)$ proves \eqref{s-eq:mu-info}.
\end{IEEEproof}

\subsection{Bounding \texorpdfstring{$C_{\mathcal I}^-$}{C I minus}}
\label{supp:coarse}
Partition the vertices into
\begin{align}
\overline U_0&=\{b_0,b_2,o\}\cup\{x_{i0},x_{i2}:0\le i<3\},\nonumber\\
\overline U_1&=\{b_1\},\qquad \overline V_i=\{a_i,x_{i1}\}.
\label{eq:coarse-partition}
\end{align}
Recall that
\begin{align*}
\overline Z_0&=(B_0^{\src},B_2^{\src},(Y_e)_{e\in\mathcal I_0}),\\
\overline Z_1&=(B_1^{\src},(Y_e)_{e\in\mathcal I_1}).
\end{align*}

For $l=0,1$, the tuple $\overline Z_l$ contains exactly the local sources and incoming transcripts of $\overline U_l$:
\[
\overline Z_l=(X_{\overline U_l}^{\src},Y_{\overline U_l,\le N}^{\mathrm{in}})
\]
up to reordering. In particular, the local source tuples are $(B_0^{\src},B_2^{\src})$ for $\overline U_0$ and $B_1^{\src}$ for $\overline U_1$.
\begin{lemma}
\label{s-lem:coarse-recovery}
We have
\begin{align}
H(\overline Z_0,\overline Z_1\mid W)&=H(A,B\mid W),\label{s-eq:coarse-complete}\\
I(\overline Z_0;\overline Z_1\mid W)&\ge H(B\mid W).\label{s-eq:coarse-common}
\end{align}
\end{lemma}

\begin{IEEEproof}
We use the five parts
\[
\overline U_0,\overline U_1,\overline V_0,\overline V_1,\overline V_2
\]
defined in \eqref{eq:coarse-partition}. Let $\mathsf T$ denote the complete bidirectional transcripts on edges joining different parts of this partition.
For $l=0,1$, the definition of $\overline Z_l$ and \eqref{s-eq:region-prefix} give
\begin{equation}
\begin{aligned}
&H(Y_{\overline U_l,\le N}^{\mathrm{send}}\mid\overline Z_l)\\
&\quad=H(Y_{\overline U_l,\le N}^{\mathrm{send}}\mid X_{\overline U_l}^{\src},Y_{\overline U_l,\le N}^{\mathrm{in}})=0.
\end{aligned}
\label{s-eq:coarse-outgoing-recovery}
\end{equation}
Every edge joining different parts connects some $\overline V_i$ to some $\overline U_l$. Its transcript toward $\overline U_l$ is contained in $\overline Z_l$, and its transcript in the opposite direction is recovered by \eqref{s-eq:coarse-outgoing-recovery}. Since $B$ is also contained in $(\overline Z_0,\overline Z_1)$,
\begin{equation}
H(B,\mathsf T\mid\overline Z_0,\overline Z_1)=0.
\label{s-eq:coarse-BT-recovery}
\end{equation}
We apply Lemma~\ref{s-lem:partition-recovery} to recover $A$ with $B$ given. Every session $h\in\mathcal A$ connects distinct terminals $a_i$ and $a_k$, which belong to distinct parts $\overline V_i$ and $\overline V_k$. The lemma therefore gives
\begin{equation}
H(A\mid B,\mathsf T)=0.
\label{s-eq:coarse-A-recovery}
\end{equation}
Consequently,
\begin{align*}
H(A,B\mid W,\overline Z_0,\overline Z_1)
&=H(A\mid W,B,\overline Z_0,\overline Z_1)\\
&=H(A\mid W,B,\mathsf T,\overline Z_0,\overline Z_1)\\
&\le H(A\mid B,\mathsf T)=0.
\end{align*}
The first equality removes the known $B$, and the second adds the recoverable $\mathsf T$, both by \eqref{s-eq:coarse-BT-recovery}. The inequality removes conditioning, and the last equality is \eqref{s-eq:coarse-A-recovery}.
Since the sources determine the transcripts, the chain rule in both orders gives
\begin{align*}
H(\overline Z_0,\overline Z_1\mid W)
&=H(A,B,\overline Z_0,\overline Z_1\mid W)\\
&=H(A,B\mid W).
\end{align*}
For the mutual-information bound, apply \eqref{s-eq:region-decoding} to $\overline U_0$ and $\overline U_1$:
\begin{equation}
H(B_0^{\dst},B_2^{\dst}\mid\overline Z_0)=0,
\qquad H(B_1^{\dst}\mid\overline Z_1)=0.
\label{s-eq:coarse-destination-recovery}
\end{equation}
Every message in $B_1^{\src}$ is requested at $b_0$ or $b_2$, so the given $(B_0^{\src},B_2^{\src})$ and recovered $(B_0^{\dst},B_2^{\dst})$ determine $B$. The tuples $B_1^{\src},B_1^{\dst}$ determine all components of $B$ except those in $W$. Therefore,
\begin{align}
H(B\mid W,\overline Z_0)
&=H(B_0^{\src},B_2^{\src},B_0^{\dst},B_2^{\dst}\mid W,\overline Z_0)\nonumber\\
&=H(B_0^{\dst},B_2^{\dst}\mid W,\overline Z_0)\nonumber\\
&\le H(B_0^{\dst},B_2^{\dst}\mid\overline Z_0)=0,
\label{s-eq:coarse-B-left}\\
H(B\mid W,\overline Z_1)
&=H(B_1^{\src},B_1^{\dst}\mid W,\overline Z_1)\nonumber\\
&=H(B_1^{\dst}\mid W,\overline Z_1)\nonumber\\
&\le H(B_1^{\dst}\mid\overline Z_1)=0.
\label{s-eq:coarse-B-right}
\end{align}
In each derivation, the second equality removes the source messages contained in $\overline Z_l$, the inequality removes $W$ from the conditioning, and the last equality follows from \eqref{s-eq:coarse-destination-recovery}.
Conditional data processing now gives
\begin{align*}
I(\overline Z_0;\overline Z_1\mid W)
&\ge I(B;\overline Z_1\mid W)\\
&=H(B\mid W)-H(B\mid W,\overline Z_1)\\
&=H(B\mid W).
\end{align*}
The inequality uses \eqref{s-eq:coarse-B-left}, and the last equality uses \eqref{s-eq:coarse-B-right}.
\end{IEEEproof}

Recall that
\[
P=H(B_0^{\src})+H(B_2^{\src})+\sum_{e\in\mathcal I_0}H(Y_e)-H(\overline Z_0).
\]

\begin{lemma}
\label{s-lem:P-tc}
The quantity $P$ satisfies $P\ge0$.
\end{lemma}

\begin{IEEEproof}
By the definition of $\overline Z_0$ and subadditivity,
\begin{align}
P&=H(B_0^{\src})+H(B_2^{\src})+\sum_{e\in\mathcal I_0}H(Y_e)\nonumber\\
&\quad-H\bigl(B_0^{\src},B_2^{\src},(Y_e)_{e\in\mathcal I_0}\bigr)\nonumber\\
&\ge0.
\label{s-eq:P-expanded}
\end{align}
\end{IEEEproof}

\begin{lemma}
\label{s-lem:coarse-bound}
We have
\begin{equation}
C_{\mathcal I}^-\ge H(A)+H(B\mid W)+P.
\label{s-eq:coarse-cost}
\end{equation}
\end{lemma}

\begin{IEEEproof}
We prove $C_{\mathcal I}^-\ge H(A)+H(B\mid W)+P$. By the definitions of $C_{\mathcal I}^-$ and $P$,
\begin{align}
C_{\mathcal I}^-
&=\sum_{e\in\mathcal I_0}H(Y_e)+\sum_{e\in\mathcal I_1}H(Y_e)\nonumber\\
&=P+H(\overline Z_0)-H(B_0^{\src})-H(B_2^{\src})
  +\sum_{e\in\mathcal I_1}H(Y_e)\nonumber\\
&=P+H(\overline Z_0)+H(B_1^{\src})
  +\sum_{e\in\mathcal I_1}H(Y_e)-H(B)\nonumber\\
&\ge P+H(\overline Z_0)+H(\overline Z_1)-H(B).
\label{s-eq:coarse-cost-expand}
\end{align}
The last equality uses source independence, which gives $H(B)=\sum_{j=0}^2H(B_j^{\src})$. The inequality follows from subadditivity applied to $\overline Z_1=(B_1^{\src},(Y_e)_{e\in\mathcal I_1})$.

The tuple $W$ is contained in both $\overline Z_0$ and $B$, so the chain rule gives
\begin{align*}
H(\overline Z_0)&=H(W,\overline Z_0)=H(W)+H(\overline Z_0\mid W),\\
H(B)&=H(W,B)=H(W)+H(B\mid W).
\end{align*}
The two $H(W)$ terms cancel in \eqref{s-eq:coarse-cost-expand}, giving
\begin{align}
C_{\mathcal I}^-
&\ge P+H(\overline Z_0\mid W)+H(\overline Z_1)-H(B\mid W)\nonumber\\
&\ge P+H(\overline Z_0\mid W)+H(\overline Z_1\mid W)-H(B\mid W).
\label{s-eq:coarse-cost-conditional}
\end{align}
The second inequality uses $H(\overline Z_1)\ge H(\overline Z_1\mid W)$. Expanding the sum of conditional entropies in \eqref{s-eq:coarse-cost-conditional}, we obtain
\begin{align*}
C_{\mathcal I}^-
&\ge P+H(\overline Z_0,\overline Z_1\mid W)
  +I(\overline Z_0;\overline Z_1\mid W)-H(B\mid W)\\
&\ge P+H(A,B\mid W)+H(B\mid W)-H(B\mid W)\\
&=P+H(A\mid B,W)+H(B\mid W)\\
&=P+H(A)+H(B\mid W).
\end{align*}
The second inequality uses \eqref{s-eq:coarse-complete} and \eqref{s-eq:coarse-common}. The first equality is the chain rule. The last equality follows because $A$ and $B$ are independent and $W$ is a subtuple of $B$.
\end{IEEEproof}

\subsection{Bounding \texorpdfstring{$\mu+P$}{mu + P}}
\label{supp:interaction}
We now prove $\mu+P\ge H(W)$. Recall from the main text that
\begin{align*}
\mathcal Q_L&=\{b_2,o\}\cup\{x_{i2}:i=0,1,2\},\\
\mathcal Q_R&=\{b_0\}\cup\{x_{i0}:i=0,1,2\},\\
\Pi&=(Y_{ox_{i0}},Y_{x_{i0}o})_{i=0}^2,\\
Z_L&=(B_2^{\src},(Y_{a_ix_{i2}})_i,(Y_{x_{i1}o})_i),\\
Z_R&=(B_0^{\src},(Y_{a_ix_{i0}})_i).
\end{align*}
The fixed schedule orders the messages in $\Pi$. The transcripts in $Z_L$ and $Z_R$ enter the corresponding vertex sets from outside $\mathcal Q_L\cup\mathcal Q_R$.
Up to reordering, $\overline Z_0=(Z_L,Z_R)$. Hence, by the definition of $P$ and subadditivity,
\begin{equation}
\begin{aligned}
P&=H(B_0^{\src})+H(B_2^{\src})+\sum_{e\in\mathcal I_0}H(Y_e)\\
&\quad-H(Z_L,Z_R)\\
&\ge H(Z_L)+H(Z_R)-H(Z_L,Z_R)\\
&=I(Z_L;Z_R).
\end{aligned}
\label{s-eq:P-LR}
\end{equation}
We first establish that $Z_1$ determines $\Pi$.
\begin{lemma}
\label{s-lem:transcript-recovery}
For the tuples $Z_1$ and $\Pi$ defined above, \mbox{$H(\Pi\mid Z_1)=0$}.
\end{lemma}

\begin{IEEEproof}
Take $\mathcal Q=\{b_1,o,x_{01},x_{11},x_{21}\}$ in Lemma~\ref{s-lem:region-simulation}. For this choice, up to reordering,
\begin{align*}
X_{\mathcal Q}^{\src}&=B_1^{\src},\\
Y_{\mathcal Q,\le N}^{\mathrm{in}}&=(Y_{x_{i0}o},Y_{a_ix_{i1}},Y_{x_{i2}o})_i,\\
Z_1&=(X_{\mathcal Q}^{\src},Y_{\mathcal Q,\le N}^{\mathrm{in}}).
\end{align*}
Since $(Y_{x_{i0}o})_i$ is contained in $Z_1$ and $o\in\mathcal Q$,
\begin{align*}
H(\Pi\mid Z_1)
&=H\bigl((Y_{ox_{i0}})_i,(Y_{x_{i0}o})_i\mid Z_1\bigr)\\
&=H\bigl((Y_{ox_{i0}})_i\mid Z_1\bigr)\\
&\le H(Y_{\mathcal Q,\le N}^{\mathrm{send}}\mid Z_1)\\
&=H\bigl(Y_{\mathcal Q,\le N}^{\mathrm{send}}\mid X_{\mathcal Q}^{\src},Y_{\mathcal Q,\le N}^{\mathrm{in}}\bigr)=0.
\end{align*}
The inequality holds because $(Y_{ox_{i0}})_i$ is among the messages sent by vertices in $\mathcal Q$, and the last equality is \eqref{s-eq:region-prefix} at $\tau=N$. Combining this with \eqref{s-eq:mu-info} yields
\begin{align}
\mu
&\ge I(W;Z_1)\nonumber\\
&=H(W)-H(W\mid Z_1)\nonumber\\
&=H(W)-H(W\mid Z_1,\Pi)\nonumber\\
&\ge H(W)-H(W\mid\Pi)\nonumber\\
&=I(W;\Pi).
\label{s-eq:mu-pi}
\end{align}
The third line uses $H(\Pi\mid Z_1)=0$, and the inequality follows from conditioning reducing entropy.
\end{IEEEproof}

\begin{lemma}
\label{s-lem:interaction}
For the tuples $W$, $\Pi$, $Z_L$, and $Z_R$ defined above,
\begin{equation}
H(W)\le I(W;\Pi)+I(Z_L;Z_R).
\label{eq:common-recovery}
\end{equation}
\end{lemma}

\begin{IEEEproof}
We prove $H(W)\le I(W;\Pi)+I(Z_L;Z_R)$. By the chain rule,
\[
H(W)=I(W;\Pi)+H(W\mid\Pi).
\]
To bound $H(W\mid\Pi)$, we first show that both $(Z_L,\Pi)$ and $(Z_R,\Pi)$ determine $W$. These pairs contain the local source messages and all incoming transcripts of $\mathcal Q_L$ and $\mathcal Q_R$, respectively. Since $b_2\in\mathcal Q_L$ and $b_0\in\mathcal Q_R$, \eqref{s-eq:region-decoding} gives
\[
H(B_2^{\dst}\mid Z_L,\Pi)=H(B_0^{\dst}\mid Z_R,\Pi)=0.
\]
The tuple $W$ is contained in both $(B_2^{\src},B_2^{\dst})$ and $(B_0^{\src},B_0^{\dst})$. Since $B_2^{\src}$ is contained in $Z_L$ and $B_0^{\src}$ in $Z_R$,
\begin{align}
H(W\mid Z_L,\Pi)&\le H(B_2^{\src},B_2^{\dst}\mid Z_L,\Pi)\nonumber\\
&=H(B_2^{\dst}\mid Z_L,\Pi)=0,\label{s-eq:interaction-W-left}\\
H(W\mid Z_R,\Pi)&\le H(B_0^{\src},B_0^{\dst}\mid Z_R,\Pi)\nonumber\\
&=H(B_0^{\dst}\mid Z_R,\Pi)=0.\label{s-eq:interaction-W-right}
\end{align}
Continuing the expansion of $H(W)$, we obtain
\begin{align}
H(W)&=I(W;\Pi)+H(W\mid\Pi)\nonumber\\
&=I(W;\Pi)+H(W\mid\Pi)-H(W\mid Z_R,\Pi)\nonumber\\
&=I(W;\Pi)+I(W;Z_R\mid\Pi)\nonumber\\
&\le I(W;\Pi)+I(Z_L;Z_R\mid\Pi).
\label{s-eq:common-recovery-conditional}
\end{align}
The second equality uses \eqref{s-eq:interaction-W-right}, and the third is the definition of conditional mutual information. The inequality is conditional data processing using \eqref{s-eq:interaction-W-left}.

It remains to show that $I(Z_L;Z_R\mid\Pi)\le I(Z_L;Z_R)$.
Write $\Pi=(\Pi_1,\ldots,\Pi_s)$ in protocol order and let $\Pi_{<k}=(\Pi_1,\ldots,\Pi_{k-1})$. Suppose that $\Pi_k$ is sent from $\mathcal Q_L$ to $\mathcal Q_R$. The tuple $Z_L$ contains $B_2^{\src}$ and the external incoming transcripts $(Y_{a_ix_{i2}},Y_{x_{i1}o})_i$. The earlier messages from $\mathcal Q_R$ are contained in $\Pi_{<k}$. Together they determine the local sources and all incoming messages of $\mathcal Q_L$ before the transmission of $\Pi_k$. Applying \eqref{s-eq:region-prefix} to $\mathcal Q_L$ therefore gives
\begin{equation}
H(\Pi_k\mid Z_L,\Pi_{<k})=0.
\label{s-eq:left-determinism}
\end{equation}
For a transmission from $\mathcal Q_R$ to $\mathcal Q_L$, $Z_R$ gives the local source tuple $B_0^{\src}$ and the external incoming transcripts $(Y_{a_ix_{i0}})_i$, while $\Pi_{<k}$ gives the earlier messages from $\mathcal Q_L$. The same reconstruction gives
\[
H(\Pi_k\mid Z_R,\Pi_{<k})=0.
\]

Consequently, for a message $\Pi_k$ sent from $\mathcal Q_L$ to $\mathcal Q_R$,
\begin{align}
I(Z_L,Z_R;\Pi_k\mid\Pi_{<k})
&=H(\Pi_k\mid\Pi_{<k})\nonumber\\
&\quad-H(\Pi_k\mid Z_L,Z_R,\Pi_{<k})\nonumber\\
&=H(\Pi_k\mid\Pi_{<k}),\label{s-eq:ext-one}\\
I(Z_L;\Pi_k\mid Z_R,\Pi_{<k})
&=H(\Pi_k\mid Z_R,\Pi_{<k})\nonumber\\
&\quad-H(\Pi_k\mid Z_L,Z_R,\Pi_{<k})\nonumber\\
&=H(\Pi_k\mid Z_R,\Pi_{<k}),\label{s-eq:int-left}\\
I(Z_R;\Pi_k\mid Z_L,\Pi_{<k})
&=H(\Pi_k\mid Z_L,\Pi_{<k})\nonumber\\
&\quad-H(\Pi_k\mid Z_L,Z_R,\Pi_{<k})\nonumber\\
&=0.\label{s-eq:int-right-zero}
\end{align}
Here \eqref{s-eq:left-determinism} and conditioning reducing entropy give
\[
0\le H(\Pi_k\mid Z_L,Z_R,\Pi_{<k})
\le H(\Pi_k\mid Z_L,\Pi_{<k})=0.
\]
Thus both conditional entropies vanish in \eqref{s-eq:int-right-zero}, as does the subtracted term in \eqref{s-eq:ext-one} and \eqref{s-eq:int-left}.
Combining these three identities with $H(\Pi_k\mid Z_R,\Pi_{<k})\le H(\Pi_k\mid\Pi_{<k})$ gives
\begin{equation}
\begin{aligned}
I(Z_L,Z_R;\Pi_k\mid\Pi_{<k})
&\ge I(Z_L;\Pi_k\mid Z_R,\Pi_{<k})\\
&\quad+I(Z_R;\Pi_k\mid Z_L,\Pi_{<k}).
\end{aligned}
\label{s-eq:message-info}
\end{equation}
The same inequality holds when $\mathcal Q_R$ sends, with $Z_L$ and $Z_R$ interchanged. Summing \eqref{s-eq:message-info} over $k$ and applying the chain rule gives
\begin{equation}
I(Z_L,Z_R;\Pi)\ge I(Z_L;\Pi\mid Z_R)+I(Z_R;\Pi\mid Z_L).
\label{s-eq:transcript-info}
\end{equation}
By the chain rule and the entropy definition of mutual information,
\begin{equation}
\begin{aligned}
&I(Z_L,Z_R;\Pi)-I(Z_L;\Pi\mid Z_R)-I(Z_R;\Pi\mid Z_L)\\
&\quad=I(Z_L;\Pi)-I(Z_L;\Pi\mid Z_R)\\
&\quad=H(Z_L)-H(Z_L\mid\Pi)\\
&\qquad-H(Z_L\mid Z_R)+H(Z_L\mid Z_R,\Pi)\\
&\quad=I(Z_L;Z_R)-I(Z_L;Z_R\mid\Pi).
\end{aligned}
\label{s-eq:interaction-identity}
\end{equation}
The first equality cancels $I(Z_R;\Pi\mid Z_L)$ after expanding $I(Z_L,Z_R;\Pi)$. The last equality groups $H(Z_L)-H(Z_L\mid Z_R)$ and $H(Z_L\mid\Pi)-H(Z_L\mid Z_R,\Pi)$. By \eqref{s-eq:transcript-info}, the left-hand side of \eqref{s-eq:interaction-identity} is nonnegative, so
\begin{equation}
I(Z_L;Z_R\mid\Pi)\le I(Z_L;Z_R).
\label{eq:interaction}
\end{equation}

Applying \eqref{eq:interaction} to \eqref{s-eq:common-recovery-conditional} gives
\begin{align*}
H(W)&\le I(W;\Pi)+I(Z_L;Z_R\mid\Pi)\\
&\le I(W;\Pi)+I(Z_L;Z_R).
\end{align*}
\end{IEEEproof}

By \eqref{s-eq:mu-pi}, \eqref{s-eq:P-LR}, and \eqref{eq:common-recovery},
\[
\mu+P\ge I(W;\Pi)+I(Z_L;Z_R)\ge H(W).
\]
The two arc-cost bounds now imply
\begin{align*}
C_-&=C_{\mathcal J}^-+C_{\mathcal I}^-\\
&\ge H(A)+H(B)+\mu+H(A)+H(B\mid W)+P\\
&=2H(A)+2H(B)-H(W)+\mu+P\\
&\ge2H(A)+2H(B)=2R.
\end{align*}
Here $H(B\mid W)=H(B)-H(W)$ because $W$ is a subtuple of $B$. This proves the reverse inequality and completes the proof of Theorem~\ref{thm:outer}.

\section{Shortest Paths in \texorpdfstring{$\Gamma_{3,3}$}{Gamma 3,3}}
\label{supp:geodesics}
Throughout this section, $d(u,v)$ denotes shortest-path distance in $\Gamma_{3,3}$ with unit edge lengths.

We verify the path property used to extend the session endpoints in Theorem~\ref{thm:all-session}. The endpoints of each path will be either $a_i,a_k$ with $i\ne k$, or $b_j,b_\ell$ with $j\ne\ell$.

\begin{lemma}
\label{lem:geodesic}
For every ordered pair of distinct vertices $(u,v)$ of $\Gamma_{3,3}$, there is a shortest path of length four between two distinct outer terminals on the same side that passes through $u$ before $v$.
\end{lemma}

\begin{IEEEproof}
Table~\ref{s-tab:geodesics} lists a path for every possible pair of vertex types. All indices range over $\{0,1,2\}$ and satisfy the conditions in the first column. Any index without an additional restriction may be chosen arbitrarily. Reversing a listed path gives the opposite order of the prescribed vertices.

Every consecutive pair in each listed path is an edge of $\Gamma_{3,3}$, and the path has four edges. Its endpoints are two distinct $a$-terminals or two distinct $b$-terminals. The bipartition is
\[
\{a_0,a_1,a_2,b_0,b_1,b_2,o\},
\qquad\{x_{ij}:0\le i,j\le2\}.
\]
The distance between vertices in the same part is even. Moreover, two distinct $a$-terminals have disjoint neighbor sets, as do two distinct $b$-terminals. Thus the distance between the listed endpoints cannot be zero or two. Together with the displayed four-edge path, this gives
\[
4\le d(a_i,a_k)\le4\quad(i\ne k),
\qquad
4\le d(b_j,b_\ell)\le4\quad(j\ne\ell).
\]
Hence every listed path is shortest. Each prescribed pair occurs along its listed path, and reversing that path supplies the other order. This proves the lemma.
\end{IEEEproof}

\begin{table}[!tb]
\caption{Shortest paths for all unordered vertex-type pairs}
\label{s-tab:geodesics}
\centering
\footnotesize
\renewcommand{\arraystretch}{1.14}
\begin{tabular}{@{}ll@{}}
\toprule
Pair and condition & Shortest path\\
\midrule
$a_i,a_k$, $i\ne k$ & $a_i-x_{ij}-o-x_{kj}-a_k$\\
$b_j,b_\ell$, $j\ne\ell$ & $b_j-x_{ij}-o-x_{i\ell}-b_\ell$\\
$a_i,b_j$, choose $k\ne i$ & $a_i-x_{ij}-b_j-x_{kj}-a_k$\\
$a_i,o$, choose $k\ne i$ & $a_i-x_{ij}-o-x_{kj}-a_k$\\
$b_j,o$, choose $\ell\ne j$ & $b_j-x_{ij}-o-x_{i\ell}-b_\ell$\\
$o,x_{ij}$, choose $k\ne i$ & $a_k-x_{kj}-o-x_{ij}-a_i$\\
$a_i,x_{k\ell}$, $k=i$, choose $r\ne i$ & $a_i-x_{i\ell}-o-x_{r\ell}-a_r$\\
$a_i,x_{k\ell}$, $k\ne i$ & $a_i-x_{i\ell}-o-x_{k\ell}-a_k$\\
$b_j,x_{k\ell}$, $\ell=j$, choose $r\ne j$ & $b_j-x_{kj}-o-x_{kr}-b_r$\\
$b_j,x_{k\ell}$, $\ell\ne j$ & $b_j-x_{kj}-o-x_{k\ell}-b_\ell$\\
$x_{ij},x_{k\ell}$, $i\ne k$ & $a_i-x_{ij}-o-x_{k\ell}-a_k$\\
$x_{ij},x_{i\ell}$, $j\ne\ell$ & $b_j-x_{ij}-o-x_{i\ell}-b_\ell$\\
\bottomrule
\end{tabular}
\end{table}

\section{Metric Packing and the Main Theorem}
\label{supp:main-reduction}
The reduction from fixed-graph communication inequalities to Theorem~\ref{thm:intro-main} follows \cite{LiuQueLiLi2026}. We give its details in the protocol model of Section~\ref{sec:model}, using the metric-packing theorems of Hirai~\cite{Hirai2010} and the communication inequalities of Yin et al.~\cite{Yin2018}.

Recall that a protocol supporting $r$ under $c$ has independent uniform source messages and a positive integer $n$ satisfying \eqref{eq:source-rate-budget} and \eqref{eq:edge-capacity-budget}. Thus $H(X_h)=\log_2|\mathcal X_h|\ge nr_h$, and its edge communication rates are
\[
f(uv)=\frac{b_{uv}+b_{vu}}{n},
\]
where $b_{uv}$ counts the bits sent from $u$ to $v$. The capacity constraint gives $f(e)\le c(e)$. For nonnegative edge lengths $\ell:E(G)\to\mathbb R_{\ge0}$, let $d_{G,\ell}$ denote shortest-path distance. Our goal is to prove, whenever $\nu(D(\mathcal H))\le2$,
\begin{equation}
\sum_{e\in E(G)}\ell(e)f(e)
\ge\sum_{h\in\mathcal H}r_hd_{G,\ell}(s_h,t_h).
\label{s-eq:metric-target}
\end{equation}
Since $f(e)\le c(e)$, this will give the distance inequalities that characterize fractional-routing feasibility. We first transfer communication inequalities from fixed graphs to $G$, then extend metric packings to real edge lengths, and finally combine the transferred inequalities using those packings.

\subsection{Transfer between graphs}
\label{supp:transfer}
Let $Q$ be a connected graph with unit edge lengths, and let $d_Q$ be its shortest-path distance. A map $\phi:V(G)\to V(Q)$ assigns each vertex $v$ of $G$ to a vertex $\phi(v)$ that will simulate it. Define
\[
m(u,v)=d_Q(\phi(u),\phi(v)).
\]
Thus $m(u,v)$ is the length of a shortest path used to forward a transmission on $uv$ between its two images. For a protocol on $Q$, write $f_Q(a)$ for the communication rate on an edge $a\in E(Q)$. In the following lemma, no edge capacities are prescribed: choose any positive integer $n$, take $r_h\le H(X_h)/n$, and divide all bit counts by the same $n$. The source messages are independent and uniform. The following transfer principle is used in \cite{LiuQueLiLi2026}; we give a direct proof for the protocol model of Section~\ref{sec:model}.

\begin{lemma}
\label{lem:transfer}
Suppose every protocol on $Q$ with independent source messages satisfies
\begin{equation}
\sum_{a\in E(Q)}f_Q(a)\ge\sum_h r_hd_Q(s_h,t_h).
\label{s-eq:transfer-assumption}
\end{equation}
Then every protocol on $G$ satisfies
\begin{equation}
\sum_{uv\in E(G)}m(u,v)f(uv)\ge\sum_h r_hm(s_h,t_h).
\label{eq:graph-metric}
\end{equation}
\end{lemma}

\begin{IEEEproof}
Fix a protocol on $G$. We construct a protocol on $Q$ supporting the same target rates for all sessions whose endpoint images are distinct, and bound its communication by the left-hand side of \eqref{eq:graph-metric}.

First consider sessions with $\phi(s_h)=\phi(t_h)$. They contribute zero to the right-hand side because $m(s_h,t_h)=0$. Fix their source messages to values of positive probability and remove their demands. Independence leaves the other source distributions unchanged, and the original protocol still decodes the remaining messages for every possible value. We retain its transmission schedule and message lengths, so the original $f(uv)$ still bounds the communication used in the construction.

For each remaining session $h$, place $X_h$ at $\phi(s_h)$ and require its recovery at $\phi(t_h)$. Vertex $x\in V(Q)$ simulates all original vertices $v$ with $\phi(v)=x$, maintaining their local sources and received messages separately. For each edge $e=uv$ of $G$, fix a shortest path $P_e$ between $\phi(u)$ and $\phi(v)$. Replace each original transmission on $uv$ by forwarding the same message along $P_e$ in the required direction. Complete this forwarding before simulating the next original transmission. If the two images coincide, use the empty path and deliver the message locally.

To verify the simulation, suppose all earlier original transmissions have been reproduced. The simulated sender has the same local source values and the same earlier received messages as the original sender. Its deterministic encoding rule therefore produces the same next message, which is delivered unchanged along $P_e$. The argument also applies to the first transmission, when there are no earlier received messages. Induction thus reproduces every original transmission. Each simulated destination can apply its original decoder and recover the required message without error. The original schedule and all paths are fixed in advance, so the new schedule is also fixed. We use the same $n$ for the constructed protocol. The unchanged source alphabets still satisfy $\log_2|\mathcal X_h|\ge nr_h$, although the number of transmissions can increase.

Write $b_e=b_{uv}+b_{vu}$ for $e=uv$, and let $b_a^Q$ count all bits sent in both directions on $a\in E(Q)$. Each bit forwarded from an original edge $e$ uses $a$ only if $a\in E(P_e)$. Hence,
\begin{equation}
 b_a^Q\le\sum_{e:\,a\in E(P_e)}b_e.
\label{s-eq:transfer-bits}
\end{equation}
Since $f_Q(a)=b_a^Q/n$, we obtain
\begin{align}
\sum_{a\in E(Q)}f_Q(a)
&\le\frac1n\sum_{a\in E(Q)}\sum_{e:\,a\in E(P_e)}b_e\nonumber\\
&=\frac1n\sum_{e\in E(G)}|E(P_e)|b_e\nonumber\\
&=\sum_{uv\in E(G)}m(u,v)f(uv).
\label{s-eq:transfer-cost}
\end{align}
The inequality uses \eqref{s-eq:transfer-bits}. The first equality interchanges the finite sums and counts the edges of each path. The last uses $|E(P_e)|=m(u,v)$ and $b_e/n=f(uv)$.

Applying \eqref{s-eq:transfer-assumption} to the constructed protocol gives
\begin{align*}
\sum_{uv\in E(G)}m(u,v)f(uv)
&\ge\sum_{a\in E(Q)}f_Q(a)\\
&\ge\sum_{h:\,\phi(s_h)\ne\phi(t_h)}r_hd_Q(\phi(s_h),\phi(t_h))\\
&=\sum_{h\in\mathcal H}r_hm(s_h,t_h).
\end{align*}
The first inequality is \eqref{s-eq:transfer-cost}, and the last equality restores the removed sessions, whose contributions are zero.
\end{IEEEproof}

\subsection{Real edge lengths in the packing theorem}
\label{supp:metric}
We first state the packing conditions needed to combine the transferred inequalities. Fix a finite connected graph $G$. For finitely many graphs $Q_\alpha\in\{K_2,K_{2,3},K_{3,3},\Gamma_{3,3}\}$ and maps $\phi_\alpha:V(G)\to V(Q_\alpha)$, write
\[
m_\alpha(u,v)=d_{Q_\alpha}(\phi_\alpha(u),\phi_\alpha(v)),
\]
where each $Q_\alpha$ has unit edge lengths. A fractional metric packing for the session pairs consists of nonnegative weights $\lambda_\alpha$ satisfying
\begin{align}
\sum_\alpha\lambda_\alpha m_\alpha(u,v)&\le\ell(uv) &&(uv\in E(G)),\nonumber\\
\sum_\alpha\lambda_\alpha m_\alpha(s_h,t_h)&=d_{G,\ell}(s_h,t_h) &&(h\in\mathcal H).
\label{eq:packing}
\end{align}
The first condition bounds the weighted distances on network edges. The second preserves the shortest-path distance of every session pair.

Theorems~1.1--1.3 of \cite{Hirai2010} provide such packings under their respective demand-graph hypotheses for cyclically even integral lengths, meaning that every cycle has even total length. We show that whenever one of these theorems applies to the demand graph, the same conditions \eqref{eq:packing} hold for nonnegative real lengths. The applicable demand-graph cases are identified in the next subsection.

First let $\ell$ be rational, and choose a positive integer $q$ such that $q\ell(e)$ is integral for every edge. Every edge length in $2q\ell$ is even, so the applicable theorem gives weights $\lambda_\alpha$ with
\begin{align*}
\sum_\alpha\lambda_\alpha m_\alpha(u,v)&\le2q\ell(uv),\\
\sum_\alpha\lambda_\alpha m_\alpha(s_h,t_h)&=d_{G,2q\ell}(s_h,t_h).
\end{align*}
Scaling all edge lengths scales every path length by the same factor. Thus $d_{G,2q\ell}=2q\,d_{G,\ell}$, and division by $2q$ gives
\begin{align*}
\sum_\alpha\frac{\lambda_\alpha}{2q}m_\alpha(u,v)&\le\ell(uv),\\
\sum_\alpha\frac{\lambda_\alpha}{2q}m_\alpha(s_h,t_h)&=d_{G,\ell}(s_h,t_h).
\end{align*}
These are \eqref{eq:packing} with weights $\lambda_\alpha/(2q)$.

Now let $\ell$ have arbitrary nonnegative real entries, and choose rational nonnegative vectors $\ell^{(t)}\to\ell$. There are finitely many maps from $V(G)$ to the four fixed graphs. List all their induced metrics as $m_1,\ldots,m_M$, independently of $t$. The rational case gives weights $\lambda_\alpha^{(t)}\ge0$; set unused weights to zero and define
\[
p^{(t)}(u,v)=\sum_{\alpha=1}^M\lambda_\alpha^{(t)}m_\alpha(u,v).
\]
The packing conditions give
\begin{align}
0\le p^{(t)}(u,v)&\le\ell^{(t)}(uv) &&(uv\in E(G)),\nonumber\\
p^{(t)}(s_h,t_h)&=d_{G,\ell^{(t)}}(s_h,t_h) &&(h\in\mathcal H).
\label{s-eq:packing-approximation}
\end{align}
Each $p^{(t)}$ satisfies the triangle inequality, since it is a nonnegative sum of induced metrics. For any $u,v$, choose a fixed path $u=v_0,\ldots,v_L=v$ in $G$. Then
\begin{align*}
0\le p^{(t)}(u,v)
&\le\sum_{j=1}^L p^{(t)}(v_{j-1},v_j)\\
&\le\sum_{j=1}^L\ell^{(t)}(v_{j-1}v_j).
\end{align*}
The last sum converges to the length of this fixed path under $\ell$, so $p^{(t)}(u,v)$ is bounded for every pair. There are finitely many pairs, hence a common subsequence converges to a vector $p$.

View each $m_\alpha$ as a vector with one coordinate for each unordered vertex pair. Their nonnegative combinations form a finitely generated cone, which is closed. Since each $p^{(t)}$ lies in this cone, its limit has a representation
\[
p(u,v)=\sum_{\alpha=1}^M\lambda_\alpha m_\alpha(u,v),
\qquad \lambda_\alpha\ge0.
\]
For each pair of vertices, shortest-path distance is the minimum of finitely many simple-path lengths, and therefore is continuous in the edge lengths. Taking limits in \eqref{s-eq:packing-approximation} yields
\begin{align*}
p(u,v)&\le\ell(uv) &&(uv\in E(G)),\\
p(s_h,t_h)&=d_{G,\ell}(s_h,t_h) &&(h\in\mathcal H).
\end{align*}
Substituting the representation of $p$ gives \eqref{eq:packing} for real edge lengths.

\subsection{Completion of the main theorem}
Recall that $D(\mathcal H)$ has the session terminals as vertices and the distinct unordered pairs $\{s_h,t_h\}$ as edges. Its matching number $\nu(D(\mathcal H))$ is the maximum number of edges with pairwise disjoint endpoints. Write $D=D(\mathcal H)$ and assume $\nu(D)\le2$.

It suffices to consider connected $G$. If a session has endpoints in different connected components, the destination receives no information about its independent source message, so zero-error decoding forces its rate to be zero. The other sessions can be treated within their components, whose demand graphs also have matching number at most two. Fix a protocol supporting $r$ under $c$ on connected $G$, a corresponding integer $n$, and arbitrary nonnegative edge lengths $\ell$.

After isolated vertices are omitted, the classification in \cite[Sec.~1]{Hirai2010} places $D$ within a graph on at most five vertices, a union of two stars, a union of a triangle and a star, or two disjoint triangles $K_3+K_3$. Theorem~1.1 there gives cut-metric packings for a union of two stars. Theorem~1.2 gives cut- and $K_{2,3}$-metric packings for the cases with at most five vertices or a triangle and a star. For $K_3+K_3$, Hirai's Theorem~1.3 gives packings using $K_2,K_{2,3},K_{3,3},\Gamma_{3,3}$; cut metrics are $K_2$-metrics. If a containing demand graph is used, its packing equalities hold on every edge of $D$, because $E(D)$ is a subset of its edge set. Thus they hold for every session pair. Section~\ref{supp:metric} extends these packings to the chosen real lengths $\ell$, yielding \eqref{eq:packing}.

For $Q=K_2,K_{2,3},K_{3,3}$, the graphs are connected, bipartite, and have diameter at most two. The entropy inequality in \cite[Corollary~1]{Yin2018} gives
\begin{align*}
n\sum_{a\in E(Q)}f_Q(a)
&\ge\sum_{uv\in E(Q)}[H(Y_{uv})+H(Y_{vu})]\\
&\ge\sum_h H(X_h)d_Q(s_h,t_h)\\
&\ge n\sum_h r_hd_Q(s_h,t_h).
\end{align*}

The first inequality bounds transcript entropy by transmitted bits, the second is the cited result, and the last uses $H(X_h)\ge nr_h$. Dividing by $n$ proves \eqref{s-eq:transfer-assumption}. For $Q=\Gamma_{3,3}$, the same inequality is \eqref{eq:gamma-rate}. Applying Lemma~\ref{lem:transfer} to each induced metric in the packing gives
\begin{equation}
\sum_{uv\in E(G)}m_\alpha(u,v)f(uv)
\ge\sum_{h\in\mathcal H}r_hm_\alpha(s_h,t_h).
\label{s-eq:transferred-packing}
\end{equation}
We now combine these inequalities using the packing weights:
\begin{align}
&\sum_{uv\in E(G)}\ell(uv)f(uv)\nonumber\\
&\quad\ge\sum_{uv\in E(G)}\left(\sum_\alpha\lambda_\alpha m_\alpha(u,v)\right)f(uv)\nonumber\\
&\quad=\sum_\alpha\lambda_\alpha\sum_{uv\in E(G)}m_\alpha(u,v)f(uv)\nonumber\\
&\quad\ge\sum_\alpha\lambda_\alpha\sum_{h\in\mathcal H}r_hm_\alpha(s_h,t_h)\nonumber\\
&\quad=\sum_{h\in\mathcal H}r_h\sum_\alpha\lambda_\alpha m_\alpha(s_h,t_h)\nonumber\\
&\quad=\sum_{h\in\mathcal H}r_hd_{G,\ell}(s_h,t_h).
\label{eq:cost-domain}
\end{align}
The first inequality uses the edge constraints in \eqref{eq:packing} and $f(uv)\ge0$. The next inequality uses \eqref{s-eq:transferred-packing} and $\lambda_\alpha\ge0$. The first two equalities interchange finite sums, and the final equality uses the session-pair equalities in \eqref{eq:packing}. This proves \eqref{s-eq:metric-target}.

Fractional-flow duality states that $r\in\mathcal C_{\MCF}(G,c,\mathcal H)$ if and only if
\begin{equation}
\sum_{e\in E(G)}\ell(e)c(e)
\ge\sum_{h\in\mathcal H}r_hd_{G,\ell}(s_h,t_h)
\quad\text{for all }\ell\ge0.
\label{s-eq:routing-duality}
\end{equation}
The capacity constraint \eqref{eq:edge-capacity-budget} gives $f(e)\le c(e)$. Therefore, \eqref{eq:cost-domain} implies
\[
\sum_e\ell(e)c(e)\ge\sum_e\ell(e)f(e)
\ge\sum_{h\in\mathcal H}r_hd_{G,\ell}(s_h,t_h).
\]
By \eqref{s-eq:routing-duality}, every supported rate vector admits a fractional routing. The fractional-routing region is closed, so taking the closure of the supported rate vectors gives $\mathcal C_{\NC}\subseteq\mathcal C_{\MCF}$. Conversely, for rational path rates $x_{h,P}$, choose a positive integer $n$ such that every $nx_{h,P}$ is an integer. Give source $h$ a binary message of length $n\sum_Px_{h,P}$, split it into portions of lengths $nx_{h,P}$, and forward each portion along $P$ in a fixed order. Each edge $e$ then carries $n\sum_h\sum_{P\ni e}x_{h,P}\le nc(e)$ bits, so the protocol supports the routed rates. Approximating feasible path rates from below by rational rates and taking closures gives $\mathcal C_{\MCF}\subseteq\mathcal C_{\NC}$. This proves Theorem~\ref{thm:intro-main}.


\begin{thebibliography}{99}

\bibitem{Ahlswede2000}
R.~Ahlswede, N.~Cai, S.-Y.~R. Li, and R.~W. Yeung, \lq\lq{}Network information flow,'' \emph{IEEE Trans. Inf. Theory}, vol.~46, no.~4, pp.~1204--1216, 2000.

\bibitem{Koetter2003}
R.~Koetter and M.~M\'edard, \lq\lq{}An algebraic approach to network coding,'' \emph{IEEE/ACM Trans. Netw.}, vol.~11, no.~5, pp.~782--795, 2003.

\bibitem{LiLi2004}
Z.~Li and B.~Li, \lq\lq{}Network coding: The case of multiple unicast sessions,'' in \emph{Proc. Annu. Allerton Conf. Commun., Control, Comput.}, Monticello, IL, USA, 2004.

\bibitem{Harvey2004}
N.~J.~A. Harvey, R.~D. Kleinberg, and A.~R. Lehman, \lq\lq{}Comparing network coding with multicommodity flow for the $k$-pairs communication problem,'' MIT Computer Science and Artificial Intelligence Laboratory, Cambridge, MA, USA, Tech. Rep. MIT-LCS-TR-964, 2004.

\bibitem{ZhangLiLi2026}
X.~Zhang, B.~Li, and Z.~Li, ``Network coding can beat routing in undirected multiple-unicast networks,'' arXiv:2610.09367, 2026.

\bibitem{BravermanHe2026}
M.~Braverman and Z.~He, ``The multiple unicast conjecture is false,'' arXiv:2610.10108, 2026.

\bibitem{Hu1963}
T.~C. Hu, \lq\lq{}Multi-commodity network flows,'' \emph{Oper. Res.}, vol.~11, no.~3, pp.~344--360, 1963.

\bibitem{Seymour1980}
P.~D. Seymour, \lq\lq{}Four-terminus flows,'' \emph{Networks}, vol.~10, no.~1, pp.~79--86, 1980.

\bibitem{OkamuraSeymour1981}
H.~Okamura and P.~D. Seymour, \lq\lq{}Multicommodity flows in planar graphs,'' \emph{J. Combin. Theory Ser. B}, vol.~31, no.~1, pp.~75--81, 1981.

\bibitem{Jain2006}
K.~Jain, V.~V. Vazirani, and G.~Yuval, \lq\lq{}On the capacity of multiple unicast sessions in undirected graphs,'' \emph{IEEE Trans. Inf. Theory}, vol.~52, no.~6, pp.~2805--2809, 2006.

\bibitem{Harvey2006}
N.~J.~A. Harvey, R.~D. Kleinberg, and A.~R. Lehman, \lq\lq{}On the capacity of information networks,'' \emph{IEEE Trans. Inf. Theory}, vol.~52, no.~6, pp.~2345--2364, 2006.

\bibitem{Yin2018}
X.~Yin, Z.~Li, Y.~Liu, and X.~Wang, \lq\lq{}A reduction approach to the multiple-unicast conjecture in network coding,'' \emph{IEEE Trans. Inf. Theory}, vol.~64, no.~6, pp.~4530--4539, 2018.

\bibitem{Karzanov1990}
A.~V. Karzanov, \lq\lq{}Sums of cuts and bipartite metrics,'' \emph{Eur. J. Combin.}, vol.~11, no.~5, pp.~473--484, 1990.

\bibitem{Karzanov1994}
A.~V. Karzanov, \lq\lq{}Paths and metrics in a planar graph with three or more holes. I. Metrics,'' \emph{J. Combin. Theory Ser. B}, vol.~60, no.~1, pp.~1--18, 1994.

\bibitem{Hirai2010}
H.~Hirai, \lq\lq{}Metric packing for $K_3+K_3$,'' \emph{Combinatorica}, vol.~30, no.~3, pp.~295--326, 2010.

\bibitem{LiuQueLiLi2026}
S.~Liu, L.~Que, Z.~Li, and B.~Li, \lq\lq{}On the multiple-unicast conjecture: Beyond cut metrics,'' arXiv:2608.06070, 2026.

\end{thebibliography}
\end{document}